\documentclass[11pt]{article}

\usepackage[top=1.0in, bottom=1.0in, left=0.95in, right=0.95in]{geometry}

\usepackage{nicefrac}
\usepackage[utf8]{inputenc} 
\usepackage[T1]{fontenc}    
\usepackage[hidelinks]{hyperref}       
\usepackage{url}            
\usepackage{booktabs}       
\usepackage{amsfonts}       
\usepackage{microtype}      
\usepackage{xcolor}         
\usepackage{amsmath}

\usepackage{amsthm}
\usepackage{graphicx}
\usepackage[font=small]{caption}
\usepackage{subcaption}
\usepackage{bm}
\usepackage[ruled,noend]{algorithm2e}
\usepackage{algpseudocode} 
\usepackage{varwidth}
\usepackage{placeins}
\usepackage{mathtools}      
\usepackage{booktabs}       
\usepackage{verbatim}
\usepackage{enumitem}
\setlist{itemsep=3pt}
\usepackage{soul}
\usepackage{scalerel}
\usepackage{cleveref}

\usepackage[numbers,sort]{natbib}
\usepackage{etoolbox}
\DeclareMathOperator*{\argmin}{argmin} 
 
\DeclareMathOperator{\E}{\mathbb{E}}
\newcommand{\lmb}{\eta}
\newcommand{\lmbt}{\eta_t}
\newcommand{\minProb}{\Delta}
\newcommand{\bfpar}[1]{%
    \vspace{3pt}%
    \noindent%
    \textbf{#1:}%
}
\newcommand{\empar}[1]{%
    \vspace{3pt}%
    \noindent%
    \emph{#1:}%
}

\newcommand{\ind}[1]{\mathbf{1}_{\{#1\}}}

\usepackage{xcolor}
\definecolor{darkgreen}{rgb}{0,0.5,0}

\usepackage[suppress]{color-edits}
\addauthor{yk}{darkgreen}
\addauthor{de}{purple}
\addauthor{et}{blue}

\newcommand{\ignore}[1]{}
\newtheorem{theorem}{Theorem}[section]
\newtheorem*{theorem*}{Theorem} 
\newtheorem{lemma}[theorem]{Lemma}

\newtheorem{proposition}[theorem]{Proposition}
\newtheorem{observation}[theorem]{Observation}

\newtheorem{claim}{Claim}
\newtheorem{definition}[theorem]{Definition}

\newtheorem*{example*}{Example}

\title{\vspace{-8pt}
Markets with Heterogeneous Agents:\\Dynamics and Survival of Bayesian vs. No-Regret Learners\thanks{An earlier version of this paper appeared in the Proceedings of EC 2025; DOI: \href{https://doi.org/10.1145/3736252.3742595}{10.1145/3736252.3742595}.} }

\author{%
  David Easley \\
  Cornell University \\
  \texttt{dae3@cornell.edu} 
  \and
  Yoav Kolumbus \\
  Cornell University \\
  \texttt{yoav.kolumbus@cornell.edu} 
  \and
  \'Eva Tardos \\
  Cornell University \\
  \texttt{eva.tardos@cornell.edu} \\
}
\date{}

\begin{document}
\maketitle

\begin{abstract}

We analyze the performance of heterogeneous learning agents in asset markets with stochastic payoffs. Our main focus is on comparing Bayesian learners and no-regret learners who compete in markets and identifying the conditions under which each approach is more effective. We formally relate the notions of survival and market dominance studied in economics and the framework of regret minimization, thereby bridging these theories. A central finding is that regret plays a key role in market selection, but low regret alone does not guarantee survival: surprisingly, an agent may achieve even logarithmic regret and yet be driven out of the market when competing against a Bayesian learner with a finite prior that assigns positive probability to the correct model. At the same time, we show that Bayesian learning is highly fragile, while no-regret learning requires less knowledge of the environment and is therefore more robust. Motivated by this contrast, we propose two simple hybrid strategies that incorporate Bayesian updates while improving robustness and adaptability to distribution shifts, taking a step toward a best-of-both-worlds learning approach. More broadly, our work contributes to the understanding of dynamics of heterogeneous learning agents and their impact on markets.

\end{abstract}

\maketitle


\section{Introduction}\label{sec:intro}

Learning how to invest in asset markets has been a central topic of study in both the economics and finance literatures, as well as in 
computer science, where algorithmic approaches to decision-making and market dynamics have received growing attention.

Learning-based algorithmic tools are now central to the functioning of many large-scale markets. In financial settings, automated trading algorithms have come to dominate both stock exchanges~\citep{securities2020staff,goldstein2023high,o2014high} and trading strategies~\citep{harris,lopezdeprado}. 

Beyond asset markets, a similar trend is evident in online advertising, where auction platforms dominated by algorithmic bidding agents allocate ad space at scale and represent a substantial share---over $1\%$---of the US GDP~\citep{marto2024rise}. Learning-based systems also operate in online retail platforms such as Amazon~\citep{brown2023competition,deng2025exploring,battalio2024cost}, where pricing, recommendations, and inventory control are largely automated. These reflect a broader shift toward markets increasingly mediated by learning agents. 

This shift raises foundational questions about the dynamics such systems induce and the outcomes they produce. From the perspective of 
decision makers managing algorithmic investment strategies in automated markets, 
a natural question arises: 
{\em which learning approach should they adopt for their algorithms in a competitive  market environment?}

Generally, players in markets aim to maximize the utility they gain from their resulting wealth. A standard utility function used in both economics and computer science is the $\log$ of the player's wealth, leading to the objective of maximizing the growth rate of wealth\footnote{The economics literature on market selection also considers different objectives based on expected discounted utility of consumption. We consider asset markets in isolation and thus focus on growth rate. 
See also Sections \ref{sec:related-work} and \ref{sec:Bayesians}.} (see references below). 

However, the existing theories of learning from  economics and from computer science propose different approaches to how learning agents  should evaluate the feedback they receive and what methods they should use to select their strategies, and the two literatures ask different questions about the market dynamics. The literature in economics primarily focuses on Bayesian learning, while the computer science literature focuses on the framework of no-regret learning. While both Bayesian and no-regret learners modify their investment decisions over time, the assumptions about the learners and the measures of success are different. 

The economics literature studies Bayesian learning, assuming that agents have a prior about the underlying process and update their beliefs over time. If these agents correctly solve their decision problems, they succeed according to their own objectives. In the case that the agents' objective is expected growth rate and they reach correct beliefs, this aligns also with the external measure of success based on their wealth share in the market. Notably, however, in this approach learners do not directly respond to changes that they experience in their wealth, but rather learn about the underlying stochastic process and best respond to it according to their current beliefs, combining their prior with their experience.

The computer science literature makes no assumptions about the underlying process, allowing it to be adversarial, and focuses on no-regret learning, measuring the resulting outcome by the regret in growth rate compared to the best single strategy in hindsight; that is, the constant investment rule an investor would have used had they known the realization of the market variables. In the case of a market with a steady underlying stochastic process that we consider, this includes a comparison to the investor investing knowing the process (although this may not be the best strategy in hindsight for a given history; see Sections \ref{sec:compare} and \ref{sec:no-regret}).
Learning agents in this literature explicitly respond to the rate of change in their wealth at each step, considering it as the utility (or loss) input to their no-regret online learning algorithm. 

We study repeated interactions in asset markets with heterogeneous learning agents: both Bayesians and regret minimizers are present in the market and aim to maximize their expected growth rate of wealth.  
A question that naturally arises in this setting is 
whether one type of agent drives the other type out of the market. That is, does an agent's wealth share converge to zero (the agent vanishes) or does it remain bounded away from zero (the agent survives)? A version of this question has been studied in the economics literature in the context of the market selection hypothesis. The results there show that agents who do not correctly maximize expected growth rate vanish in the presence of agents who do correctly maximize expected growth rate. See the related work section for further discussion and references to this literature.   
It is important to note that both the computer science and economics literatures on this topic treat the asset market as ``perfectly competitive''; that is, although prices are endogenous (equilibrium prices evolve as wealths change over time), individual agents are ``price takers'' that do not take into account any impact of their decisions on equilibrium prices.

To understand which learning approaches are more useful in this competitive setting, we ask the market selection question for  markets populated with a greater variety of learning agents. Some are Bayesian optimizers, some are Bayesians who make small mistakes in the set of models they consider or in updating their beliefs, and some use no-regret learning. A key question that arises is: what is the meaning or role of regret in this context? Regret-minimizing agents achieve low regret in the complete asset market we analyze, but we show that if there is a Bayesian who makes correct updates and whose prior on models puts positive probability on the correct model, then regret minimizers vanish. That is, surprisingly, {\em successfully reaching low regret does not guarantee survival}. However, we also show that {\em Bayesian learning is fragile}; if the Bayesian does not have the correct model within the support of its prior, or if the Bayesian does not correctly balance the prior and likelihood functions in updating, then the Bayesian almost surely vanishes relative to the no-regret learner, even if the Bayesian makes arbitrarily small errors. 

With these insights in hand, we propose two simple and practical learning methods that combine the fast convergence of Bayesian updates with greater robustness.

In the first method, the key idea is to introduce a small regularization into the Bayesian update, so that, initially, the weight on incorrect models still decays exponentially fast, as in standard Bayesian learning, but convergence slows enough to prevent any hypothesis from being completely ruled out. We show that this simple modification preserves the constant-regret guarantee in stationary settings, while also maintaining responsiveness to distribution shifts within a given model class, with regret growing only logarithmically in time when such shifts occur.   
Intuitively, once the weight on a model becomes small, the regularization keeps it ``alive,'' preserving the ability to recover if that model later becomes relevant again. In contrast, we show that standard Bayesian learning can suffer regret linear in $T$ even from a single distribution shift. While adaptive variants of no-regret learning exist \citep{hazan2009efficient,kozat2007universal,hazan2016introduction}, they still incur regret that grows with time even in the absence of shifts, and thus may not survive in competition against players investing according to the correct distribution, or even against Bayesians.

The second method addresses the fragility of Bayesian learning under prior misspecification. The idea is to follow the Bayesian strategy while tracking its realized performance relative to a no-regret algorithm simulated in parallel. If the Bayesian prior is misspecified, this performance gap grows, and the algorithm switches to the no-regret strategy, which has already been updated on the observed data so far. We show that when the Bayesian prior is well specified, the algorithm stays with the Bayesian strategy with high probability and achieves constant expected regret overall. In the worst case, it guarantees sublinear regret by inheriting the regret bound of the no-regret algorithm, up to the switching threshold. 
The main challenge is to choose the switching rule so that false switches are rare: even when the prior is correct, random fluctuations in the stochastic process may temporarily make the Bayesian learner appear to lag behind the no-regret benchmark, although switching away from the correct Bayesian strategy would lead to higher regret in the long run.

\subsection{Summary of Main Results}\label{sec:results-overview}
To the best of our knowledge, this work is the first to analyze the dynamics of no-regret learners in competition with agents using other learning paradigms; to study the conditions under which no-regret learners survive in investment settings (in contrast to only bounding their regret level); and to establish a connection between the frameworks of regret minimization and Bayesian learning in asset markets. 
Our analysis yields a clear characterization of the relationship between regret rates, wealth shares, and long-term survival in such competitive environments. 

Beyond the theoretical contribution, our results provide practical insights for selecting learning strategies in market-facing applications. 
We show that competition imposes a high-stakes tradeoff between faster and more robust learning. 
The right approach depends on the extent to which campaign managers can access and trust information about the environment and the strategies of competitors. Bayesian methods, as suggested by economic theory, are powerful when the model class is accurate and known, and the update implementation is precise, but risky otherwise. No-regret learning methods, by contrast, are useful and robust when such knowledge is limited. 
The new methods we propose combine the strengths of both approaches, highlighting  regularization and hybrid strategies as practical tools for balancing this tradeoff.

\vspace{8pt}
\noindent
\textbf{Main Theorems:} Before turning to related work, we conclude this part of the introduction with an informal summary of our main theorems. 

\empar{Regret and Survival}  
Theorem~\ref{thm:survival-by-finite-regret-gap} characterizes the regret condition for survival in a competitive asset market: an agent survives if and only if its regret remains bounded by an additive constant relative to every competitor.

Remarkably, if two players both have similar regret asymptotics with respect to the best strategy in hindsight (e.g., both have $O(\log T)$ regret) but with different coefficients, the agent with the larger constant almost surely vanishes from the market.  

\empar{Regret of the Optimal Investment Strategy}  
Regret with respect to the best strategy in hindsight is a strong benchmark; as we discuss in Section~\ref{sec:compare}, due to stochasticity in market dynamics, even the best implementable strategy (i.e., one that does not depend on future outcomes, but uses the true probabilities of states) has positive regret. In Theorem~\ref{thm:regret-of-q}, in Section~\ref{sec:no-regret}, we calculate this expected regret level, showing that it is constant and depends only on the number of assets in the market.
This provides a natural benchmark for comparison with agents that invest optimally according to the true distribution of states.

\empar{Survival against Perfect Bayesians}  
Theorem~\ref{thm:regret-of-a-Bayesian} shows that perfect Bayesian learners with a finite support have constant expected regret, and that their regret difference relative to an agent playing $q$ is also pathwise bounded.
Intuitively, because a perfect Bayesian maintains a regret level close to the optimal implementable benchmark (investing according to the true distribution), any learner whose regret grows at a faster rate cannot survive against the Bayesian.

\empar{Bayesians with Inaccurate Priors}  
Theorem~\ref{thm:wrong-Bayesians-vanish} shows that a Bayesian agent with the correct state distribution not within its support suffers linear regret, and, as a result, vanishes from the market in competition with any no-regret learner.  
This theorem is followed by an analysis of the expected survival times of Bayesians with inaccurate priors, as a function of the size of the errors in the models considered in their priors and the regret rate of the best competitor.  

\empar{Bayesians with Inaccurate Updates}  
In Section~\ref{sec:noisy-Bayesian}, we analyze a scenario where a Bayesian learner performs ``trembling hand'' updates, making small zero-mean errors in the weight given to the current prior and the new data at each step.  
Theorem~\ref{thm:wrong-update-Bayesians-vanish} demonstrates that even such tiny errors break Bayesian learning, causing the learner to fail to converge and incur linear regret.  
As a result, any no-regret (or other) learner that converges to the correct model at any rate outperforms such an inaccurate Bayesian, driving it out of the market.  

\empar{Robust Bayesian Updates}  
In Section~\ref{sec:robust-Bayes-update}, we study two methods for making Bayesian learning more robust. Proposition~\ref{thm:dist-shift-linear-regret-proposition} shows that standard Bayesian learning can incur linear regret under distribution shifts even within the support of the Bayesian prior. Theorem~\ref{thm:robust-Bayesian} shows that a regularized Bayesian update achieves constant regret when there are no distribution shifts and logarithmic regret in $T$ when shifts occur. Theorem~\ref{thm:hybrid-Bayes-No-regret} considers misspecified priors and shows that our switching algorithm achieves constant expected regret when the prior is well specified and sublinear worst-case regret.

\vspace{8pt}
\noindent
\textbf{Roadmap:}
The paper is structured as follows. After a review of related work,   
Section~\ref{sec:model} presents the market model and preliminary analysis, discussing relative wealths and the role of relative entropy in characterizing survival under stationary investment rules.  
Section~\ref{sec:compare} reviews known results from no-regret learning and analyzes the relation between regret and wealth shares, showing how regret can be used to compare learners.  
Section~\ref{sec:Bayesians} restates in the language of our market setting several known analyses of Bayesian learning that are important for our comparison of Bayesians, no-regret learners, and imperfect Bayesians.   
Section~\ref{sec:no-regret} analyzes the regret of a player who knows the correct stochastic model, and gives survival conditions in competition with a perfect Bayesian.  
Section~\ref{sec:imperfect-Bayesians} studies imperfect Bayesians with small mistakes in their priors or update process.  
Section~\ref{sec:robust-Bayes-update} 
presents methods  
for making Bayesian learning more robust under distribution shifts and prior misspecification. 
Section~\ref{sec:simulations} presents simulation results for several scenarios of competition between no-regret learners, perfect Bayesians, and imperfect Bayesians, to complement our analysis and provide further intuition.  
Finally, Section~\ref{sec:discussion} discusses the modeling choices and the scope and interpretation of the analysis, and Section~\ref{sec:Concln} concludes with broader implications of the results. Formal proofs are deferred to Appendix~\ref{sec:appendix-proofs}, and further background on the Arrow securities model is given in Appendix~\ref{sec:appendix-Arrow-securities}.

\subsection{Related Work}\label{sec:related-work}

\bfpar{No-regret learning}
No-regret learning is a fundamental concept in the study of learning in games, tracing back to early work in game theory ~\citep{blackwell1956analog,brown1951iterative,hannan1957lapproximation,robinson1951iterative} and subsequently developed in influential papers on no-regret algorithms ~\citep{DBLP:journals/siamcomp/AuerCFS02,blum2008regret,foster1997calibrated,fudenberg1995consistency,hart2000simple,kalai2005efficient}. See ~\cite{cesa2006prediction,hart2013simple}, and~\cite{slivkins2019introduction} for broad introductions. The dynamics of no-regret learning algorithms in repeated games have been extensively studied in various domains, including repeated auctions ~\citep{aggarwal2024randomized,bichler2023convergence,daskalakis2016learning,guo2022no,feng2020convergence,deng2022nash,kolumbus2022auctions,kolumbus2024paying}, budgeted auctions ~\citep{balseiro2019learning,feng2024strategic,fikioris2023liquid,fikioris2024learning,fikioris2026robust,lucier2024autobidders}, portfolio selection, as we discuss below  ~\citep{DBLP:journals/ml/BlumK99,gofer2016lower,hazan2015online,li2012expected}, repeated contracting ~\citep{collina2024repeated,guruganesh2024contracting,zhu2022sample}, bilateral trade ~\citep{cesa2021regret,cesa2023bilateral}, and implications of no-regret algorithms on user incentives in general games ~\citep{kolumbus2022and}. In all these works on learning dynamics, either a single learner is considered, or there are multiple learners where it is assumed that all are regret minimizers.

When optimizing growth rate of wealth over multiple interactions, the resulting  growth rate (the change in the $\log$ of wealth) is an additive function over the periods, so classical no-regret learning applies. In the context of investment portfolio selection, ~\cite{DBLP:journals/ml/BlumK99} and~\cite{hazan2015online} introduced  algorithms with worst-case regret bounds of $O(\log(T))$. See also ~\cite{cesa2006prediction}, Chapter 10 and ~\cite{hazan2016introduction} Chapter 4 for a  discussion and analysis of investment strategies achieving this regret bound.  
These studies emphasize regret as the key metric, and in that, implicitly assume that if the time-average of the regret vanishes, then the agent is considered successful. However, we show that in a competitive market, an agent may achieve vanishing regret yet still be driven out of the market, with wealth share converging to zero. 

\bfpar{Bayesian learning} In economics, it is standard to view learning as a by-product of expected utility maximization. That is, the decision-making agent's objective is not learning per se, but making good decisions as evaluated by expected utility. It is well known, though not often written formally, that expected utility theory, and its subjective axiomatic foundation, as in~\cite{savage}, imply Bayesian learning. So the focus on Bayes is not surprising. 

There is also a long tradition in economics and finance of arguing that the reason financial markets price assets correctly is that the forces of market selection drive out irrational traders, leaving rational traders to determine prices. This argument is historically most closely associated with~\cite{alchain1950}, \cite{friedman1953}, and~\cite{fama1965}. 
Here, rationality is interpreted as having a statistically correct model of the stochastic process that determines the returns on assets and optimizing according to it.\footnote{No-regret learning, in comparison, can be viewed from this perspective as a form of bounded rationality.} As this notion of rationality implies Bayes, the argument is often extended to argue that markets select for Bayesians and against all other learning procedures that do not at least asymptotically behave like Bayes. The intuition for the `market selection' argument is simple: if rational and irrational traders disagree, then rational traders take money away from irrational ones in the events that rational traders consider to be likely. The argument then goes on to note that as rational traders compute conditional probabilities correctly, their view of likely events is at least asymptotically correct. So they survive and irrational traders do not survive. This argument has several important qualifications, most importantly, complete markets and the possibility of eventually-correct beliefs for rational traders. The market selection hypothesis, and qualifications to it, have subsequently been investigated by~\cite{blumeeasley1992}, \cite{sandroni2000},  and~\cite{blumeeasley2006}. 

There is also a literature in finance and operations research on wealth accumulation for investors who maximize the expected growth rate of their wealth. The investment rule that implements this objective is known as the Kelly rule~\citep{Kelly1956}. There is a long-standing controversy on the applicability of this rule in consumption-investment decisions, see~\cite{Samuelson1971fallacy} or~\cite{Hens2009} for a treatment of wealth dynamics with a variety of objectives including expected growth rate maximization. 

Interestingly, the computer science literature on no-regret learning and the economics literature on Bayesian learning in asset markets differ not only in their questions and methods, but are also almost entirely disjoint, with very few cross-references between them. In this work, we establish a formal connection between these different approaches to learning in markets. We show that there is a surprisingly simple relationship between the entropy of the learner---a key notion in the Bayesian learning literature---and the regret of the learner in these markets. We then analyze a heterogeneous group of Bayesian learners and regret minimizers interacting in a competitive market, and study the implications of these different learning approaches for survival in the market.

\bfpar{Dynamics in markets and investment games}
A rich body of literature explores the dynamics of various market models. Prominent examples include studies on the t\^atonnement price dynamics~\citep{cheung2018amortized,cole2008fast,walras1900elements}, proportional response dynamics in Fisher markets~\citep{cheung2018dynamics,kolumbus2023asynchronous,zhang2011proportional}, dynamics in exchange economies~\citep{branzei2021exchange,branzei2021proportional} and the von Neumann model of expanding economies~\citep{branzei2018universal,branzei2025tit}, repeated Cournot competition games~\citep{bischi2008learning,farrell1989renegotiation,kolumbus2022and}, and dynamics in pricing games~\citep{banchio2023adaptive,hartline2024regulation,hartline2025regulation}.
These studies focus on the joint dynamics of agents, analyzing their convergence properties, equilibria, and implications for overall outcomes. While many of these dynamics are (or can be framed as) distributed learning processes, these works assume homogeneous learning behavior among agents and do not address settings where learners use fundamentally different learning approaches. 
Other research on investment games focuses on equilibrium outcomes and interventions designed to incentivize investment~\citep{babaioff2022optimal,HalacKremerWinter2020raising,jackson2021systemic}, which do not involve learning.

\bfpar{Games with heterogeneous learners}
The present work studies interactions between heterogeneous learning agents and their implications for market dynamics and outcomes. The literature on learning in games contains very few studies that analyze dynamics between learners using different strategies. The most closely related work to ours is~\cite{blumeeasley1992}, which shows that Bayesian learners outperform heuristics based on imitation or search in investment games. Another line of work on repeated games with players using different types of reasoning
concerns Stackelberg games where an optimizer seeks the best strategies against an opponent who commits to using a fixed no-regret algorithm in the game~\citep{arunachaleswaran2024algorithmic,arunachaleswaran2024learning,arunachaleswaran2024pareto,BravermanMaoSchneiderWeinberg2018selling,cai2023selling,guruganesh2024contracting,mansour2022strategizing}. This differs significantly from our study where we do not assume that any player is strategically optimizing as a Stackelberg leader.

\section{Model and Preliminaries}\label{sec:model}
We use the market structure as in \cite{blumeeasley1992} in our analysis of wealth and price dynamics. 
At each time in discrete steps indexed by $t=1,2,\dots$, after portfolios are chosen, a state $s_t$ is realized. There is a finite number of states $s\in \{1, \dots, S\}$ distributed i.i.d. with probability $q=(q_1, \dots, q_S)$ where $q_s \geq \Delta > 0$ for all $s$. There are $K$ assets with state contingent payoffs described by a $K \times S$ matrix of asset payoffs $A$.  We assume that the asset structure is {\bf complete}; that is, $A$ has rank $S$. So there must be at least $S$ assets, and if there are more than $S$ assets, then $K-S$ of them are redundant. Thus, we restrict attention to $K=S$ assets. 

It is convenient to analyze wealth dynamics using an equivalent set of assets known as
Arrow securities. These assets are indexed by states: one unit of asset $s$ purchased at time $t$ pays one unit of wealth if $s_t=s$ and zero otherwise. Any complete asset structure can be represented by a simple transformation of Arrow securities. 
This asset market model, originally introduced by \cite{Arrow}, is the standard general equilibrium model of an economy with a complete set of assets, see \cite{MWG} Chapter 19. We use Arrow securities only as a modeling convenience. Agents may invest in the actual assets in an economy with an arbitrary, but complete, asset market. We discuss this equivalence in Appendix \ref{sec:appendix-Arrow-securities}.

There are $N$ agents indexed by $n=1, \dots, N$. The wealth of agent $n$ after time $t$ is $w^n_t$. Agents begin with initial wealths $w^n_0>0$ for agent $n$. These initial wealths are normalized to sum to one.\footnote{This normalization is done only for convenience. We analyze relative wealths. The actual size of the economy (the aggregate wealth) does not play a role in relative wealth dynamics and may follow an arbitrary stochastic process, as long as it is bounded away from zero and bounded from above.} 
Let $\alpha^n_{st}$ be the fraction of agent $n$'s wealth invested in asset $s$ at time $t$.  

In a market equilibrium, the price of asset $s$ must be such that the aggregate demand for asset $s$ equals the aggregate supply of one. So the price is given by 
\begin{equation}
p_{st} = \sum_n \alpha^n_{st}w^n_{t-1}.
\end{equation}
Note that asset prices sum to one as the aggregate wealth is one.

Using this structure, we can describe the evolution of wealths and asset prices as functions of the agents' investments. The step-by-step stochastic process of agent $n$'s wealth is

\begin{equation}\label{eq:wealth-evolution}
w^n_t = \prod_{s=1}^S  \bigg(\frac{\alpha^n_{st}}{p_{st}}\bigg)^{\mathbf{1}_{s_t = s}} w^n_{t-1},
\end{equation}
where $\mathbf{1}_{s_t = s}$ is the indicator function taking on the value $1$ if the state realized at time $t$ is state $s$, i.e., $s_t=s$, and equals zero otherwise. Equivalently, the wealth at time $t > 0$ is 
\begin{equation}\label{eq:wealth-equation}
w^n_{t}   = \prod_{\tau=1}^{t} \prod_{s=1}^S  \bigg(\frac{\alpha^n_{s\tau}}{p_{s\tau}}\bigg)^{\mathbf{1}_{s_\tau = s}} w^n_0.
\end{equation}

Directly analyzing individual wealth processes is potentially complex as they depend on prices which are endogenous. Instead, we analyze relative wealths. Let $r^{nm}_t$ be the ratio of the wealth of agent $n$ to the wealth of agent $m$ at time $t$. Prices cancel, so wealth ratios follow
\begin{align}\label{eq:wealth-share}
	r^{nm}_{t} &=  \prod_{s=1}^S  \bigg(\frac{\alpha^n_{st}}{\alpha^m_{st}}\bigg)^{\mathbf{1}_{s_t = s}} r^{nm}_{t-1},\\
	r^{nm}_{T}   &= \prod_{t=1}^{T} \prod_{s=1}^S \bigg (\frac{\alpha^n_{st}}{\alpha^m_{st}}\bigg)^{\mathbf{1}_{s_t = s}} r^{nm}_0.
\end{align}
Taking the logarithm of both sides yields
\begin{equation}\label{eq:logshares}
	\log(r^{nm}_{T}) = \sum_{t=1}^{T}\sum_{s=1} ^S \mathbf{1}_{s_t = s} \log (\frac{\alpha^n_{st}}{\alpha^m_{st}}) + \log(r^{nm}_{0}). 
\end{equation}

To provide some insight into how wealth ratios evolve, it is useful to first consider a special case in which investment sequences are constant over time: $\alpha^n_t=\alpha^n$ for all $t$. With constant investments, dividing both sides of equation (\ref{eq:logshares}) by $T$, taking the limit as $T\to \infty$, omitting the constant term $\log(r_0^{nm})$ that vanishes in the limit, and applying the Law of Large Numbers yields\footnote{Note that if $\alpha^n_{st}$ were a randomized strategy with expectation $\alpha^n$ we could not simply take the expectation, but by Jensen's inequality, we would get a wealth ratio bounded by the entropy difference $\lim_{T\to \infty} T^{-1} \log(r^{nm}_{T}) \leq I_q(\alpha^m) - I_q(\alpha^n)$.}

\begin{equation}\label{eq:wealthshares}
  \lim_{T\to \infty} T^{-1} \log(r^{nm}_{T}) = \sum_{s=1}^S q_s \log(\frac{\alpha^n_s}{\alpha^m_s}). 
\end{equation}

To interpret equation (\ref{eq:wealthshares}) it is useful to define the relative entropy (a.k.a. KL-divergence) of investment rule $\alpha$ relative to the probability distribution $q$ as

\begin{equation}
    I_q(\alpha) = \sum_{s=1}^S q_s \log(\frac{q_s}{\alpha_s}).
\end{equation}
Then, rewriting equation (\ref{eq:wealthshares}) using relative entropies we have 

\begin{align}
\label{eq:fixed invext entropy}
    \lim_{T\to \infty} T^{-1} \log(r^{nm}_{T}) = I_q(\alpha^m) - I_q(\alpha^n). 
\end{align}

Thus, in the constant investment rules environment, the fate of agent $n$ depends on $I_q(\alpha^n)$ and $I_q(\alpha^m)$ for all other traders $m$. For example, if $n$ has the unique lowest relative entropy then $n$'s wealth converges to $1$ and all other agents' wealth converge to $0$. Entropy is minimized at $\alpha=q$, so if a single agent invests according to $q$, that agent's wealth share converges to $1$.

In the next two sections, we consider Bayesian learners and regret minimizers. Their investment rules are not constant, so we need to generalize this analysis to allow for non-constant rules. 

Before considering specific learning strategies, it is useful to note that due to the multiplicative nature of wealth evolution (see Equation (\ref{eq:wealth-evolution})), any learning strategy that assigns zero weight to any state risks losing all wealth, with no possibility of recovery; formally:

\begin{definition}\label{def:survive-and-vanish}
    An agent $n$ {\em survives} in the market if there exists a constant market share $c > 0$ such that \ 
    $
    \liminf_{T \rightarrow \infty} w^n_T > c
    $ \ almost surely.  
    We say that agent $n$ {\em vanishes} if $n$'s  wealth converges to zero in the limit. That is, \ 
    $
    \lim_{T \rightarrow \infty} 
    w^n_T = 0$ \ almost surely. 
\end{definition}

\begin{observation}
    Let $\{\alpha^n_t\}_{t=0}^\infty$ be a sequence of investment profiles for agent $n$. If there exists a state $s \in S$ for which $\alpha^n_{st} = 0$ infinitely often, then the agent $n$ vanishes from the market.
\end{observation}

This simple observation immediately rules out certain learning approaches, such as those that treat asset selection as pure actions. It demonstrates that any reasonable learning strategy must be restricted to investment profiles with full support over the assets. 
Specifically, all learning strategies we consider will assign at least a small constant probability $\minProb$ to each state.

\section{Comparing Learners}\label{sec:compare}
As we noted above, an investor who knows the probabilities $q_s$   maximizes the expected growth rate of their wealth at each time $t$ by using the investment rule $\alpha^n_{t}=q$. Next, consider an agent $n$  with a sequence of investment profiles $\alpha_t^n$. We will first compare the wealth of this agent with that of an agent who uses the investment rule $\alpha^n_{t}=q$, and then compare different types of learning agents.

As we saw in the previous section in Equation (\ref{eq:logshares}), the log of the wealth ratio of an agent at time $T$ compared to an investor using the optimal $\alpha^n_{t}=q$ rule can be expressed as

\begin{equation}\label{eq:investment return}
\sum_{t=1}^{T}\sum_{s=1}^S \mathbf{1}_{s_t=s} \log \alpha^n_{st} -\sum_{t=1}^{T}\sum_{s=1}^S \mathbf{1}_{s_t=s} \log q_s + \log r^{nm}_{0}. 
\end{equation}

The main goal of this paper is to compare different theories on how to best learn to invest. One theory is Bayesian learning, assuming that the investor $n$ has a Bayesian prior over the possible asset return probabilities. At each step, a Bayesian learner uses the expected return probability implied by their beliefs, updating their beliefs according to the observed states using Bayes' rule. One can also use multi-armed bandit learning methods given such a prior of possible asset return distributions, and have the agent invest using one of the distributions in its prior at each step, aiming to learn which one is the best. A different theory would consider the maximization problem of Equation (\ref{eq:investment return}) as a concave maximization problem and use no-regret learning algorithms to optimize the value. 

From the latter two perspectives, if the difference between the two sums in Equation (\ref{eq:investment return}) is negative, this is a form of regret\footnote{\label{footnote:regret-wrt-q}This is regret only with respect to $q$; while $q$ is optimal in expectation, it may not be the best fixed strategy in hindsight for a given sequence of state realizations. We expand on this point later; see Theorem  \ref{thm:regret-of-q}.} against a fixed strategy $\alpha=q$. First, consider an investor with a fixed number of possible models for $q$, $\{\theta^1, \ldots, \theta^K\}$. This formulation of the problem would allow us to consider an investor using a classical multi-arm bandit approach, such as UCB (see, e.g., \cite{slivkins2019introduction}): the investor can invest using one of the $\theta^k$ as investment rules, and use no-regret methods to update their belief on the expected value of each of the options. Assuming the correct distribution $q$ is among the options considered, we can express the gap in value between any one choice $\theta^k$ and the true distribution using Equation (\ref{eq:fixed invext entropy}) as 
\begin{equation}\label{eq:MAP-gap}
   I_q(\theta^k)-I_q(q)=I_q(\theta^k).
\end{equation}

Classical learning algorithms, such as UCB, then guarantee regret against the optimal strategy $q$ depending inversely on the above gap that is at most 
\begin{equation}
 O(\log T) \cdot\sum_{k:\theta^k\neq q}\frac{1}{I_q(\theta^k)}.  
\end{equation}
In the next section, we consider a Bayesian investor who not only has a finite list of possible $q$ vectors, but also has a prior over these models, and uses Bayes' rule to update their prior given the outcomes seen in previous rounds.  

Alternatively, we can use learning without a list of options or a prior, and  hence no longer depend on the assumption that the correct probability vector is one of the options considered. Let us define the utility of an investor at time $t$ with investment strategy $\alpha$ as 
\begin{equation}
    U_t(\alpha)=\sum_s \mathbf{1}_{s_t = s} \log \alpha_{st}.
\end{equation}
$U_t(\alpha)$ is the rate of change in the investor's relative wealth in this period compared to the investor using $\alpha=q$.
Notice that this is a concave function of $\alpha$ for any outcome of the random variables. For an investor $n$ with investment sequence $\alpha^n_{1:T}=(\alpha^n_1, 
\ldots, \alpha^n_T)$, the value the investor is aiming to optimize is
\begin{equation}
\sum_{t=1}^{T}U_t(\alpha^n_t).
\end{equation}
Using these definitions, the classical regret from no-regret learning for an investment sequence $\alpha_{1:T}^n$ at time $T$ is defined as 
\begin{equation}
    R_T(\alpha_{1:T}^n) = \max_{\alpha\in \Delta_S} \sum_{t=1}^T U_t(\alpha)-\sum_{t=1}^T U_t(\alpha^n_t).
\end{equation}

Note that the regret of $\alpha^n_{1:T}$ is at least as large as its regret against the investor with the optimal fixed rule $\alpha=q$. 
The maximum fixed rule $\alpha$ used in the definition is the best fixed investment rule with hindsight, having seen the outcome of all choices and random state realizations, which is not an implementable investment rule even for an investor who knows the probabilities. So even the investor with knowledge of the probabilities and using $\alpha^n_t = q$ for all $t$ will have some regret based on the outcome of random events (see Section \ref{sec:no-regret} and Theorem \ref{thm:regret-of-q} for a detailed analysis).

Next, we reconsider two investors, $n$ and $m$, with investment sequences $\alpha_{1:T}^n$ and $\alpha_{1:T}^m$, respectively. By Equation (\ref{eq:logshares}), we can express their time-$T$ wealth-shares using regret notation as
 \begin{align}
 \label{eq:logshare-as-regret}
	\log(r^{nm}_{T})   & = \sum_{t=1}^{T} U_t(\alpha^n_t)-\sum_{t=1}^{T} U_t(\alpha^m_t) + \log(r^{nm}_{0})\\ \notag
    & = R_T(\alpha_{1:T}^m)-R_T(\alpha_{1:T}^n)+\log(r^{nm}_{0}).
\end{align}
Note that this relation is not asymptotic, but applies for any time $T$.  
So, given the initial wealth shares, the agents’ regret levels specify their resulting wealths.

Despite this very strong benchmark for regret, and without any prior, the convex optimization literature offers very strong guarantees on expected regret. For any concave function $U_t(\alpha)$, a simple gradient-descent based algorithm guarantees regret at most $\sqrt{T}$ over $T$ periods, while for the $\log$-based utility function (that is Exp-Convex) considered here, second-order optimization methods guarantee a $O(\log T)$ regret bound~\citep{DBLP:journals/ml/BlumK99} (see also~\cite{cesa2006prediction} Theorem 3.3 or~\cite{hazan2016introduction} Theorem 4.4.)

The central question we ask is:  
{\em
What learning algorithms are better in ensuring that the player's wealth remains high, or at the very least, does not vanish? How do Bayesian learners and no-regret convex optimizers compare in terms of relative wealth? 
}
Using Equation (\ref{eq:logshare-as-regret}), we get the following result, connecting regret and survival in the market.
The proof is given in Appendix \ref{sec:appendix-proofs}. 

\begin{theorem}\label{thm:survival-by-finite-regret-gap}
    Denote the regret levels of agents, $n$ and $m$, by $R^n(T)$  and $R^m(T)$ respectively. If $ \ \lim_{T \rightarrow \infty} R^n(T) - R^m(T) = \infty$ a.s., then agent $n$ vanishes from the market. 
    Agent $n$ survives if and only if for every agent $m \neq n$ it a.s. holds that $\limsup_{T \rightarrow \infty} R^n(T) - R^m(T) < \infty$.
\end{theorem}

Recall that the entropy of the investment rule using the correct distribution of states is $I_q(q)=0$. 
Also, recall from equation (\ref{eq:MAP-gap}) that the gap in expected reward between this and a different investment rule $q'$ is the entropy $I_q(q')$. This shows that the wrong fixed strategy will have linear regret $\Theta(I_q(q')T)$. 
Learners with $O(\sqrt{T})$ do better and learners with only $O(\log T)$ regret do even better. Note that the constant in the regret rate will matter. The investor with regret higher by only a constant factor will also vanish! But how well do all these learners do against a Bayesian learner?

\section{Bayesian Learners}\label{sec:Bayesians}
Next, we consider the standard Bayesian learning analysis. We restate and summarize in this section known results for completeness and to provide a framework for our comparison of Bayesian and no-regret learning.

Consider a Bayesian watching the process on states in the market. Suppose that the Bayesian does not know the probability $q$, and instead considers a finite number of possibilities for $q$: models $\theta^1, \dots, \theta^K$. We assume that $\theta^k_s \geq \minProb$ for all $k$ and $s$. Survival with Bayesian models of this type has   also been studied in \cite{blumeeasley1992,blumeeasley2006}. See the remark at the end of this section for a connection to results on Bayesian learners with a continuum prior. The Bayesian has a prior $\lambda_0=(\lambda^1_0, \dots, \lambda^K_0)$ with $\lambda^k_0 > 0$ for all $k$. 
Let $n_t^s$ be the number of times that state $s$ occurs by time $t$. The posterior probability of model $\theta^k$ is, by Bayes rule,
\begin{align}
\lambda^k_t &= \frac{\lambda^k_0\prod_s (\theta^k_s)^{n_t^s}}{\sum_l \lambda^l_0\prod_s (\theta^l_s)^{n_t^s}}.
\end{align}
Thus, for any models $k$ and $l$
\begin{align}
\frac{\lambda^k_t }{\lambda^l_t } &= \frac{\lambda^k_0\prod_s (\theta^k_s)^{n_t^s}}{\lambda^l_0\prod_s (\theta^l_s)^{n_t^s}},\\
\log\frac{\lambda^k_t }{\lambda^l_t } &= 
\log\frac{\lambda^k_0}{\lambda^l_0} + \sum_s n_t^s \log(\theta^k_s) - \sum_s n_t^s \log(\theta^l_s).\label{eq:log-odds-by-n}
\end{align}
Using the law of large numbers and the definition of relative entropy, we have, almost surely, 
\begin{align}\label{eq: bayes limit}
\lim_t \frac{1}{t} \log\frac{\lambda^k_t }{\lambda^l_t } = I_q(\theta^l) - I_q(\theta^k).
\end{align}

If model $k$ is correct, $\theta^k = q$, then $I_q(\theta^k)=0$, so for any other model $\theta^l$ with $ I_q(\theta^l) \neq 0$ we have $\lambda^l_t \to 0$ and $\lambda^k_t \to 1$ almost surely as $t \rightarrow \infty$. Equation (\ref{eq: bayes limit}) shows that the log of posterior probabilities converges linearly in time, and thus, that the ratio of posterior probabilities converges at an exponential rate. 

The Bayesian learns the true model if it is one of the models considered. Alternatively, suppose that no model considered by the Bayesian is correct, but that there is a unique model with minimum entropy relative to the true model; i.e., there exists $k^*$ such that $0<I_q(\theta^{k^*})<I_q(\theta^k)$ for all $k\neq k^*$. Then it follows from equation (\ref{eq: bayes limit}) that $\lambda_t^{k^*} \to 1$ almost surely as $t \rightarrow \infty$. 
These convergence results can be summarized in the following proposition. 
\begin{proposition}\label{thm:bayesian-convergence}
    Let $\Theta$ denote the support of a Bayesian learner's prior, and let 
    $
    \theta^* = \argmin_{\theta \in \Theta} I_q(\theta)
    $. 
    The Bayesian's posterior probability on $\theta^*$ converges to $1$ almost surely. 
\end{proposition}

The Bayesian's time $t$ predicted probability on states is $q_t = \sum_k \lambda^k_{t-1} \theta^k$. If one of the models the Bayesian considers is correct, then this predicted probability converges at an exponential rate to $q$ .

\subsection{An Investment Rule for the Bayesian}

We next generate an optimal investment rule for the Bayesian from the predicted probability on states. A Bayesian trader $n$ with predicted probability $q^n_t$ who wants to maximize the expected growth rate of wealth selects at time $t$ an investment rule that solves: 
\begin{equation} \label{eq:bays-investment-rule-w}    
\max_{\alpha^n_{t}\geq 0} \sum_s {q^n_{st}} \log\frac{w^n_{t}}{w^n_{t-1}},  \ \quad 
s.t.  \sum_s \alpha^n_{st} = 1 \notag.  
\end{equation}
Using the definition of wealth, we have 
\begin{equation}\label{eq:bays-investment-rule-p}
\max_{\alpha^n_{t}\geq 0} \sum_s {q^n_{st}} \log\frac{\alpha^n_{st}}{p_{st}},  \ \quad
s.t. \sum_s \alpha^n_{st} = 1. \notag
\end{equation}
\noindent The solution is clearly $\alpha^n_t = q^n_t$. The Bayesian agent should ``bet its beliefs.'' 

\vspace{5pt}
\noindent
{\em Remark:} An alternative foundation for this investment rule can be derived from a Bayesian whose objective is to maximize its expectation of discounted sum of log utility of consumption. Suppose the agent begins with some wealth $w^n_0$, prior $\lambda^n_0$ on models, and discount factor $0 < \delta < 1$. If there is a complete set of Arrow securities at each time, this agent will optimally invest a fraction $\delta$ of its wealth at each time, deriving utility from consuming the rest, and use the investment rule $q^n_t$. Assuming that all agents have the same discount factor yields results identical to those for an economy where agents maximize expected growth rate of wealth.

\subsection{Wealth Dynamics for the Bayesian}
A Bayesian agent that maximizes the expected growth rate of its wealth uses the investment rule $\alpha^n_t=q^n_t$ at each time $t$ (i.e., betting its beliefs). We consider the evolution of the wealth of the Bayesian relative to any other trader $m$ who uses an investment rule that does not depend on future states, i.e., its time-$t$ value depends only on state realizations through time $t-1$.   
The log ratio of their wealths is, from Section \ref{sec:model},
  \begin{align}\label{eq:log-ratio-of-wealth}
	\log(r^{nm}_{T})   &= \sum_{t=1}^{T}\sum_{s=1} ^S \mathbf{1}_{s_t = s} \log \left(\frac{\alpha^n_{st}}{q_s}\right) \\ \notag
    & - 
    \sum_{t=1}^{T}\sum_{s=1} ^S \mathbf{1}_{s_t = s} \log \left(\frac{\alpha^m_{st}}{q_s}\right) +
    \log(r^{nm}_{0}). 
\end{align}
Note that 
\begin{align}
\sum_{t=1}^{T} \ 
&\E\big[ 
\log \big(\frac{\alpha^n_{st}}{q_s}\big)
\big]
-
\E\big[ \log \big(\frac{\alpha^m_{st}}{q_s}\big)
\big]\\
&=
\sum_{t=1}^{T} I_q(\alpha^m_t)-\sum_{t=1}^{T} I_q(\alpha^n_t).
\end{align}

If one of the models that the Bayesian considers is the correct model, $q$, then $I_q(\alpha^n_t)$ converges exponentially to $0$ almost surely. So in this case,  $\sum_{t=1}^{T} I_q(\alpha^n_t)$ is finite almost surely. The relative entropy $I_q(\alpha^m_t)$ is at least $0$ for any investment rule and is $0$ if $\alpha^m_t=q$ for all $t$. Thus, the sums of relative entropies are at least equal to the negative of the finite sum of relative entropies for the Bayesian. 

Subtracting the means from the sums in Equation (\ref{eq:log-ratio-of-wealth}) yields a martingale, and the arguments in \cite{blumeeasley1992} imply the following conclusion:
\begin{proposition}\label{thm:Bayesians-suvive-against-q}
    Let $r^{nm}_T$ be the wealth ratio of a Bayesian learner indexed by $n$ competing against an agent indexed by $m$ playing $\alpha^m_t=q$ for all $t$. Then 
    $
    \liminf_{T \rightarrow \infty} r^{nm}_T > 0
    $
    almost surely; i.e., the Bayesian survives.
\end{proposition}

Note that if the correct model is not in the set of models considered by the Bayesian, then the predicted distribution converges to the model closest in relative entropy to $q$, and the sum of relative entropies diverges. In this case, if some investor, $m$, uses  $\alpha^m_t = q$ for all $t$, then the Bayesian's wealth converges almost surely to $0$. 

\vspace{5pt}
\noindent
{\em Remark:}
In this paper we consider Bayesian learners with a finite set of models. The Bayesian analysis can be extended to priors with a continuum of models in their support, see \cite{clarkebarron1990}. In this case, learning occurs at rate $O((\log T)/T)$, where the constant depends on the dimension of the set of models considered. We observe that this convergence rate yields a regret of $O(\log T)$ at time $T$. This implies, by our Theorems \ref{thm:survival-by-finite-regret-gap} and \ref{thm:regret-of-q}, that a Bayesian learner with continuum prior does not survive in a market when competing with an investor who knows the true underlying process, or with a Bayesian with a finite support. The result about the effect of the dimension of the support of the Bayesian's prior is known in the economics literature, see \cite{blumeeasley2006}, but the approach there does not use a measure of regret. 

Thus, such a Bayesian will have regret of the same order of magnitude as a no-regret learner. In fact, there is an interesting connection between Bayesian learning with a convex continuum support and no-regret learning using the exponential weights algorithm. One can think of the classical exponential-weights algorithm (as used in Theorem 3.3 of \cite{cesa2006prediction}  or Theorem 4.4 of \cite{hazan2016introduction}, when choosing their parameters $\eta$, $\alpha$, respectively, as $\eta=\alpha=1$) as doing Bayesian learning with a uniform prior over the convex set of distributions on states. It is interesting to note that for a $d$-dimensional convex prior, \cite{clarkebarron1990} show convergence at the rate of $\frac{d}{2}(\log T)/T$, resulting in regret $\frac{d}{2} \log T$ (though they do not discuss regret). By contrast, no-regret learning results prove a $d \log T$ regret bound. While this second bound is a factor of 2 worse, note that the no-regret result is worst case: it does compare to the best single investment strategy in hindsight, but does not assume a steady process underlying the market. Since the two algorithms follow the same investment strategy, the no-regret results prove that Bayesian learning over a convex continuum prior of dimension $d$ also guarantees a worst-case regret of at most $d \log T$ against the best fixed strategy in hindsight, even if the Bayesian's  assumption of a steady underlying process turns out incorrect.

\section{The Regret of Bayesian Learners}\label{sec:no-regret}

In this section we further analyze the relationship between an agent's  regret  and long-term survival in market dynamics, and use this analysis to derive the regret level of Bayesian learners. This relationship depends on the competition. In particular, we saw in Theorem \ref{thm:survival-by-finite-regret-gap} that the learners who survive (have a positive expected wealth level in the limit) are only those who maintain a constant regret gap relative to the best competitor in the market. 

As mentioned in Section \ref{sec:compare}, using the state distribution $q$ as the investment rule (if $q$ were known)  maximizes the expected growth rate among constant strategies. However, even this strategy incurs some regret, as it generally differs from the best strategy in hindsight. We now ask:  what is the regret level of using the true state distribution $q$ as the investment rule? This regret level represents the best attainable expected regret. Thus, by Theorem \ref{thm:survival-by-finite-regret-gap}, evaluating this regret level would serve as the benchmark for determining which regret levels ensure survival (see Definition \ref{def:survive-and-vanish}) against any competitors who do not have information about future state realizations. 
To set the ground, let us first consider the following definition:

\begin{definition}
Fix an arbitrary history of state realizations $s_1,\dots,s_T$ and denote the empirical distribution of this history by $\hat{q}$, where $\hat{q}_s = \frac{1}{T}\sum_{t=1}^T \ind{s_t=s}$. A {\em ``magic agent,''} indexed also by $\hat{q}$, is an agent playing the (eventual) empirical distribution in every step. That is,  $\alpha^{\hat{q}}_{st} = \hat{q}_s$ for all $t \in [T]$.     
\end{definition}

Note that this agent uses information regarding future realizations of random states, which is not available to real agents in a stochastic environment. Next, we derive the regret level of using the correct distribution of states $q$ as the investment strategy.
The proof is given in Appendix \ref{sec:appendix-proofs}; it uses our analysis from Section~\ref{sec:compare} and a result from information theory on the relative entropy between the underlying distribution $q$ and the empirical realization of states.  

\begin{theorem}\label{thm:regret-of-q}
    An agent using the state distribution $q$ as the investment strategy for all $t$ has  expected regret bounded by a constant $R$ that depends only on the number of states $S$. 
\end{theorem}

Our analyses above and in the preceding sections lead to the following characterization of the regret of Bayesian learners and the survival conditions for other learning agents competing with them in the market. The proof is given in Appendix \ref{sec:appendix-proofs}.

\begin{theorem}\label{thm:regret-of-a-Bayesian}
    A Bayesian learner with a finite-support prior that includes the correct state distribution $q$ satisfies the following for any market parameters:
    \begin{enumerate}
        \item Its expected regret is bounded by a constant at all times.
        \item Its pathwise regret difference relative to an agent playing $q$ is bounded from above.
    \end{enumerate}
\end{theorem}

The next lemma provides an expression for the regret of a player as a function of the entropy relative to strategy $q$, using our notation $R$ from Theorem \ref{thm:regret-of-q} for the expected regret of strategy $q$. This expression will be useful in the following Section on imperfect Bayesians.

\begin{lemma}\label{thm:expected-regret-lemma}
    The expected regret of a constant strategy sequence (that is, $\alpha^n_{1:T}$ such that $\alpha^n_t =  \alpha^n$ for all $t$) is given by 
    $
    \E[R^T(\alpha^n_{1:T})] = T \cdot I_q(\alpha^n) + R.
    $ 
\end{lemma}

\section{Imperfect Bayesians}\label{sec:imperfect-Bayesians}
In the preceding sections, we saw that a perfect Bayesian is optimal in the sense that it survives almost surely in the market and drives out any player with regret increasing over time. In this section, we explore a scenario where a Bayesian learner suffers small errors, either in its set of prior distributions or in its update rule. We find that, while perfect Bayesians are a powerful benchmark, Bayesian learning is also very fragile; for example, it is key that the correct probability distribution is one of the models they consider. Specifically, we show that an imperfect Bayesian would eventually vanish from the market, either against the state distribution $q$ or in competition with any player with sublinear regret. This holds even if the Bayesian's errors are very small and have zero mean.
 
\subsection{Bayesian Learning with Inaccurate Priors}\label{sec:noisy-priors-Bayesian}

As mentioned in Section \ref{sec:Bayesians}, a Bayesian learner with $q$ not within the support of its prior converges to using the best strategy within its prior (i.e., the one closest to $q$ in relative entropy), denoted here by $q'$. 
This convergence property leads to the following result. 
\begin{theorem}\label{thm:wrong-Bayesians-vanish}
    A Bayesian agent with the state distribution $q$ not in its support incurs regret 
    linear in $T$, and vanishes from the market in competition against any no-regret learner.   
\end{theorem}
Note that this contrasts with the case of Theorem \ref{thm:regret-of-a-Bayesian}, where a perfect Bayesian dominates the market in competition with any learners with regret increasing in time.
The proof is in Appendix \ref{sec:appendix-proofs}.

\vspace{5pt}
\noindent
{\em Remark:}
The possible strategies that the Bayesian can play span the entire convex hull of the prior. 
Bayesians with a finite prior can (and do) play convex combinations of the prior during the dynamic, but still, they always converge to the vertices---even when a more profitable strategy lies in the interior. 
In this sense, this convergence property of Bayesians is both a strength (if they have an accurate prior) and a weakness (if they do not). 
\vspace{5pt}

\bfpar{Typical Survival Time}
A question that arises now is what happens when errors are small.  
In the following analysis, we consider ``$\epsilon$-inaccurate Bayesians'' who have in their prior a close approximation of $q$, but still with some small error. Specifically, the total variation distance between the best strategy in the prior and the state distribution $q$ satisfies $TV(q', q) = \epsilon$ for some $\epsilon > 0$. 

On the one hand, we know that such an agent asymptotically fails to survive against a regret-minimizing agent for any error $\epsilon$. On the other hand, when $\epsilon$ is small, the inaccurate Bayesian will initially converge very quickly to a distribution close to the correct state distribution, and may capture a significant share of the market value during the early stages. 

Our goal now is to analyze the survival time during which the inaccurate Bayesian may still maintain a significant market share, and to understand how this survival time is related to the level of error $\epsilon$. This, of course, also depends on the best competitor and how fast they converge.

To estimate this relation and provide an upper bound on the survival time of an inaccurate Bayesian learner, we consider a player using $q'$ with $TV(q', q) = \epsilon$ (having already converged to this distribution), competing against a second player playing a dynamic strategy with expected regret level increasing as $f(t)$ that is sub-linear in $t$.

First, we would like to express the first player's error in terms of entropy (i.e., KL-divergence with $q$). We can bound the entropy using Pinsker's inequality~\citep{pinsker1964information} for a lower bound and its inverse for an upper bound, were the latter holds for finite distributions with full support, as in our case. We have (recall that $\Delta = \min_s q_s$),
\[
2\epsilon^2 \leq 
I_q(q')
\leq \frac{2}{\Delta} \epsilon^2.
\]
Using Lemma \ref{thm:expected-regret-lemma}, this can be translated into the regret of the inaccurate agent playing strategy $q'$ where $R$ is the regret of an investor using the true state distribution $q$ (see Theorem \ref{thm:regret-of-q}): 
\[
I_q(q')
= \frac{1}{T} \big(\mathbb{E}[R^{T}(q'_{1:T})] - R \big).
\]

To maximize the survival time, we pick the distribution $q'$ to be the one with the smallest regret, which is when $I_q(q')$ is the smallest possible: $I_q(q') = 2\epsilon^2$. So we have\footnote{If the inaccurate agent uses the worst $q'$ with error $\epsilon$, the regret calculation is similar, but with $\epsilon$ rescaled by a factor of $1/\sqrt{\Delta}$, using the inverse of Pinsker's inequality: 
$2 \epsilon^2 \tau/\Delta  + R$. The smallest-regret  $q'$ is used for the bound.} 
$\mathbb{E}[R^{T}(q'_{1:T})] = 2\epsilon^2 T + R$. 
By Equation (\ref{eq:logshare-as-regret}), the regret difference between two players is equal to the log of their wealth ratios  plus a constant. By comparing regret levels between the agents, we obtain a bound on the typical survival time $\tau$ up to which the inaccurate player holds, in expectation, more than half the market value; beyond this point, their expected share is less and continues decaying to zero:\footnote{An alternative question one may ask is: how accurate does my prior need to be to allow retaining a share of the market for $T$ steps against the competitor? For this, one can invert the equation and get: $\epsilon = \sqrt{(f(T) - R)/T}$.}
\[
f(\tau) = 2\epsilon^2 \tau + R.
\]

\vspace{5pt}
\noindent
{\em Remark:} Note that $f(t)$ is the competitor's actual regret. Regret bounds for known learning algorithms (e.g., bandit algorithms such as UCB (see, e.g.,   \cite{slivkins2019introduction}) or gradient-descent and second-order methods like those described in \cite{hazan2016introduction}) are typically worst-case bounds. Estimating the expected regret (or the most likely one) in game dynamics in general, or specifically in our investment scenario, is an interesting open problem.
\vspace{5pt}

For the special cases that the competitor has constant regret (for example, an accurate Bayesian), or the competitor has a $\log T$ or $\sqrt{T}$ regret level, we get the following results. 
\begin{observation}\label{obs:survival-time-gainst-constant-regret}
    The expected survival time of a Bayesian learner with $\epsilon$-inaccurate prior against any player with constant regret (e.g., a perfect Bayesian) is $O\big(\frac{1}{\epsilon^2}\big)$. 
\end{observation}

For a competitor with logarithmic regret, we get an equation of the form $\tau = c_1 \cdot \exp(c_2 \epsilon^2\tau)$. For small errors $\epsilon$, by expanding the exponent, we have the following bound:
\begin{observation}\label{obs:survival-time-gainst-log-regret}
    The expected survival time of a Bayesian learner with $\epsilon$-inaccurate prior against any player with logarithmic regret (e.g., a no-regret convex optimizer) is $O\big(\frac{1}{\epsilon^3}\big)$.   
\end{observation}

Finally, in competition against a learner with regret $f(T) = \sqrt{T}$ the survival time is longer:
\begin{observation}\label{obs:survival-time-gainst-sqrt-regret}
    The expected survival time of a Bayesian learner with $\epsilon$-inaccurate prior against any player with $\sqrt{T}$ regret is $O\big(\frac{1}{\epsilon^4}\big)$. 
\end{observation}

The interpretation of this analysis naturally depends on the time scale of interest. When $T$ is large, which is our main focus (e.g., when trade occurs at high frequency or if investments are long-term), only long-term survival matters. In this case, a Bayesian with an inaccurate prior will eventually vanish in competition with any learner who converges to the truth (e.g., a regret minimizer). However, when transient dynamics are also relevant, it is possible that a Bayesian with a prior that is inaccurate but relatively close to the true distribution may retain a significant market share for some time before ultimately vanishing. See Section~\ref{sec:simulations} for an example.

\subsection{Bayesian Learning with Noisy Updates}\label{sec:noisy-Bayesian}
Next, we consider a different type of imperfect Bayesian learner that does have $q$ in its prior, but in every step performs slightly inaccurate ``trembling hand'' updates. Here as well, we demonstrate that Bayesian learning is fragile, even when the correct distribution lies in the support. To model this, we define a noisy Bayesian learner as one who at each step, either slightly overweighs the current observation or slightly overweighs its current prior, such that, in expectation, both the data and the prior receive the correct weights in every step (i.e., the errors in weight have zero mean).

For concreteness, consider the following scenario of a learner attempting to learn a distribution of states. Suppose that there are two states, $s_t = 0$ with a fixed probability $q \in (0,1)$, and $s_t = 1$ otherwise. The learner considers two models: $\theta_a = q$, which is the correct model, and $\theta_b \neq q$, with $\theta_b < 1$. The log-likelihood is given by
\begin{align}
    L(s_t) = 
    (1 - s_t) \log\Big(\frac{\theta_a}{\theta_b}\Big) + 
    s_t \log\Big(\frac{1 - \theta_a}{1 - \theta_b}\Big).
\end{align}

Now suppose that in every step $t$ the Bayesian learner performs ``$\lmb$-noisy updates'' where it over- or under-weights the data compared to the prior with a small excess weight $\lmbt$, where $\lmbt$ has $0$ mean. Specifically, for a parameter $\lmb > 0$, $\lmbt= \lmb$ with probability $1/2$ and $\eta_t = -\lmb$ otherwise. We find that even a tiny zero-mean error has a significant impact, essentially breaking the learning process.

\begin{theorem}\label{thm:wrong-update-Bayesians-vanish}
    For any $\lmb > 0$, the Bayesian learner with $\lmb$-noisy updates does not converge, and therefore incurs regret linear in $T$. 
\end{theorem}

The idea of the proof is to show that, with high probability, the learner has a systematic drift toward over-weighing new observations, despite the symmetry of the errors around zero. As a result, the learner’s beliefs fail to converge and continue to fluctuate in response to recent random events, leading to linear regret. The full proof appears in Appendix~\ref{sec:appendix-proofs}.

\section{Robust Bayesian Updates}\label{sec:robust-Bayes-update}
In the analysis above, we saw that the fast convergence of Bayesian learners is both a strength and a weakness. On the one hand, when a Bayesian converges to a correct model, it achieves constant regret and survives even against a fully informed competitor. On the other hand, if it converges to an incorrect model, it suffers linear regret (relative to the best fixed strategy in hindsight) and vanishes from the market. 
This strong convergence also becomes problematic when the data-generating process undergoes distribution shifts. 
In this section, we develop two methods to utilize Bayesian updates and achieve constant regret in a more robust manner. The first method addresses distribution shifts within a given model class. The second addresses the problem of misspecified priors by combining Bayesian learning with a safe switching strategy to a no-regret fallback.

\subsection{Distribution Shifts}\label{sec:dist-shift}
We begin with a simple modification of Bayesian updating for environments in which the data-generating distribution may shift over time. The method we propose preserves the fast convergence of Bayesian learning at early stages, while regularizing the update enough to improve robustness to shifts. 
We first show that, although a standard Bayesian learner may eventually adapt after a distribution shift, its regret relative to the true sequence of distributions can still be linear in the total time $T$.

\begin{proposition}\label{thm:dist-shift-linear-regret-proposition}
    Consider a market in which the data-generating distribution $q$ may shift over the course of $T$ steps among a finite set $Q$ of distributions. 
    A Bayesian learner with a finite prior that assigns positive probability to each model in $Q$ may incur regret linear in $T$ compared to a benchmark that knows the correct sequence of distributions in advance, even if there is only a single distribution shift during the $T$ steps. 
\end{proposition}

The proof is in Appendix~\ref{sec:appendix-proofs}. 
The Bayesian's fast convergence is a useful feature if the Bayesian puts positive prior probability on a fixed, correct model of the data-generating process and for asymptotic results this is all that matters; the prior is otherwise irrelevant. But if the data generating process shifts over time the Bayesian’s posterior at the (unknown) time of a shift becomes the prior. So intermediate-term properties matter. In the intermediate term, a Bayesian who has nearly converged to a past model learns slowly, as the log odds of the new model to the past one are extreme and the expected shift in these log odds is the fixed relative entropy between the two models. With repeated shifts, the Bayesian thus spends a large amount of time near models that are incorrect and accumulates regret that can be linear. 

By contrast, there are no-regret learning algorithms that are known to handle distribution shifts more gracefully \citep{hazan2009efficient,kozat2007universal}. In particular, for the portfolio selection problem we consider, such algorithms can achieve regret that is logarithmic in $T$ and linear in the number $n$ of distribution shifts (see also \cite{hazan2016introduction}, Sections 10.3 and 10.4). However, as we have seen, logarithmic regret is not sufficient for survival when competing against agents who invest according to the true distribution, or even against Bayesian learners. 

The main observation we make here is that, while the exponential convergence rate that Bayesians achieve is sufficient to obtain constant regret---and thus ensures long-run survival against fully informed competitors---it is not necessary. We propose a simple and natural method that combines advantages of both approaches, taking a step towards a best-of-both-worlds learning strategy. It guarantees constant regret in stationary settings, like standard Bayesian learning, thus outperforming no-regret learners in such cases. At the same time, it guarantees logarithmic regret under distribution shifts, surviving against no-regret learners in settings where adaptability is essential. 

Technically, our method stays close to Bayesian learning. 
We use a regularized version of the Bayesian update rule that initially behaves like standard Bayesian learning, but slows convergence just enough to keep alternative hypotheses alive. This enables fast recovery after shifts. 

\vspace{8pt}
\noindent
\textbf{Robust Bayesian Update:}
The \emph{Robust Bayesian Update} process we propose uses a simple and efficient modification of standard Bayesian learning. Initially, the learner sets uniform prior weights $\lambda^1_0, \dots, \lambda^K_0$. At each time $t \geq 0$, the learner uses the current weights $\lambda_t$ as its investment strategy. Then, after observing the outcome, it updates beliefs with an added regularization term $\epsilon_t = t^{-2}$:
\begin{equation}\label{eq:robust-Bayes-update}
\lambda_{t} \rightarrow \text{ (Bayes update) } \rightarrow \tilde{\lambda}_{t+1} \rightarrow \lambda_{t+1} = \frac{\tilde{\lambda}_{t+1} + \epsilon_t}{1 + K \epsilon_t}.    
\end{equation}
where $K$ is the number of hypotheses.  
That is, using the observed outcome at time~$t$, the learner applies the standard Bayesian update to obtain tentative weights~$\tilde{\lambda}_{t+1}$. It then adds regularization and normalizes to get weights $\lambda_{t+1}$ in the simplex.   
This simple process satisfies the following guarantee:

\begin{theorem}\label{thm:robust-Bayesian}
 Let $Q$ be the set of models in the learner's support. Suppose an adversary selects a sequence of intervals $\big(\{T_i, q_i\}\big)_{i=1}^n$, where $q_i \in Q$ and $T_i$ is the duration of the $i$-th interval during which $q_i$ is the data-generating distribution. Let $T = \sum_{i=1}^n T_i$ denote the total time. Then the regret of the Robust Bayesian Update, relative to a benchmark of fixed investment strategies in each interval, knowing the intervals in advance, is $O(n \log T)$ with high probability. In particular, when $n$ is constant, the regret is logarithmic in $T$. 

 Moreover, if there are no distribution shifts (i.e., when $n=1$), the regret remains bounded by a constant with high probability.
\end{theorem} 
\vspace{3pt}

The intuition behind the proof is that when the weight of an incorrect model is large, it decays exponentially fast, as in the case of a standard Bayesian learner. When it is small---approaching the scale of the regularization term---it no longer contributes significantly to regret. Thus, the regret under a stationary process remains bounded. After a distribution shift, all model weights are of at least the size of the regularization level, which prevents suppressing the alternatives at an exponential rate and keeps them ``alive.'' This is then followed by exponential convergence, leading to regret that grows only logarithmically in $T$. The proof of the theorem is in Appendix~\ref{sec:appendix-proofs}.

\subsection{Bayes to No-Regret Switching Strategy}
The simple regularization method described above resolves the adaptability issue of Bayesian updating from Proposition~\ref{thm:dist-shift-linear-regret-proposition} and yields a regret guarantee for distribution shifts within the class $Q$. It does not, however, provide a general adversarial-regret guarantee, and in particular, it does not address misspecification. If the data-generating distribution lies outside the support, then the learner may incur linear regret. 

We now consider a hybrid strategy that provides robustness both when the prior is well specified and when it is misspecified. The algorithm starts with the Bayesian strategy and switches permanently to a no-regret strategy if the realized performance of the Bayesian learner falls sufficiently far behind. 
The main challenge is to ensure that false switches resulting from random fluctuations in the market process remain sufficiently rare so as to preserve the constant regret of Bayesian learning when the prior is well specified, while still guaranteeing sublinear regret in the worst case.

The one-period utility is
$ 
U_t(\alpha) = \log \alpha_{s_t} 
$, 
where $s_t$ is the  realized state. After each round, an investor can evaluate both the Bayesian action and no-regret action on that outcome.  
The switching rule uses the difference between the two.

\vspace{5pt}
\noindent
\textbf{The switching strategy:}
Recall that portfolios are in the set
$
A = \{\alpha\in\Delta_S:\alpha_s\geq \Delta \text{ for all } s\}
$. 
Let $\alpha_t^B$ be the Bayesian investment based on a prior over models $\Theta=\{\theta^1,\dots,\theta^K\}$, and let $\alpha_t^N$ be the action at time $t$ of a 
no-regret algorithm run in parallel on the observed outcomes.  
Define
\begin{equation}
d_t = U_t(\alpha_t^N) - U_t(\alpha_t^B), \qquad D_t = \sum_{t'=1}^t d_{t'}.
\end{equation}
Since both regrets are measured against the same action set, the hindsight benchmark cancels, and
\begin{equation}
D_t = R_t(\alpha^B_{1:t}) - R_t(\alpha^N_{1:t}),
\end{equation}
where
\begin{equation}
R_t(\alpha_{1:t})
= 
\max_{\alpha\in  A}\sum_{t'=1}^t U_{t'}(\alpha)
- 
\sum_{t'=1}^t U_{t'}(\alpha_{t'}).
\end{equation}
Also, since $U_t(\alpha)\in[\log \Delta,0]$, we have $|d_t|\leq L$, with $L=\log(1/\Delta)$.

We set a dynamic threshold
\begin{equation}
b_t
=  
L\sqrt{2 t \log \left( \frac{\pi^2 t^2}{6 \delta \lambda_{\min}} \right)},
\end{equation}
where 
$\lambda_0^1,\dots,\lambda_0^K$ are the Bayesian prior weights,  $\lambda_{\min} = \min_{k \in [K]} \lambda_0^k$, and $\delta > 0$ is a parameter to be determined later.   
This threshold grows at rate $\tilde O(\sqrt t)$. If $D_t \geq b_t$, the algorithm switches permanently to the no-regret strategy.

\vspace{3pt}
\noindent
\textbf{Algorithm (Bayes to No-Regret Switch).}
Fix a finite Bayesian prior and a no-regret algorithm.
Initialize in \emph{Bayes} mode. At each step $t$:
\begin{enumerate}
    \item Compute $\alpha_t^B$ and $\alpha_t^N$ from the current history.
    \item Play $\alpha^{\mathrm{sw}}_t = \alpha_t^B$ if still in \emph{Bayes} mode, and otherwise play $\alpha^{\mathrm{sw}}_t = \alpha_t^N$.
    \item Observe state $s_t$ and update $D_t$. If $D_t \geq b_t$, switch permanently to \emph{No-regret}.
\end{enumerate}

\begin{theorem}[Bayes to No-Regret Switch]\label{thm:hybrid-Bayes-No-regret}
Fix a horizon $T$ and set $\delta=1/T$.
\begin{enumerate}
    \item If the true distribution $q$ is in the support of the Bayesian prior, then with probability at least $1-1/T$ the algorithm does not switch. 
    Moreover,
    $$
    \E[R_T(\alpha^{\mathrm{sw}}_{1:T})]
    \leq
    \E[R_T(\alpha^B_{1:T})] + L.
    $$
    In particular, by Theorem \ref{thm:regret-of-a-Bayesian}, its expected regret is bounded by a constant.

    \item For every realization,
    $$
    R_T(\alpha^{\mathrm{sw}}_{1:T})
    \leq
    R_T(\alpha^N_{1:T}) + b_T + L.
    $$
    Therefore, if $\alpha^N$ is a no-regret strategy with $O(\sqrt{T})$ regret or better, then
    $$
    R_T(\alpha^{\mathrm{sw}}_{1:T}) = \tilde O(\sqrt T).
    $$
\end{enumerate}
\end{theorem}

The idea of the proof is that, first, the switching rule is based on the cumulative realized utility gap $D_t$, which is exactly the regret difference between the Bayesian and no-regret strategies. Thus, when $D_t$ crosses the threshold, the switching algorithm's regret can lag behind the no-regret strategy by at most the threshold at the switching time. This gives the worst-case regret guarantee. Second, when the Bayesian prior is well specified, false switches are rare, because a well-specified Bayesian cannot accumulate regret quickly enough to cross the threshold. The proof uses the prior predictive distribution as a device to bound the probability of such a switch, using the observation that under this mixture probability measure,  $D_t$ is a supermartingale. 
Thus, under correct specification, with high probability the algorithm never switches and behaves like the Bayesian learner, preserving its constant-regret guarantee. In the  case that the Bayesian model is misspecified, the algorithm switches once the Bayesian strategy falls sufficiently far behind, and then inherits the regret guarantee of the no-regret algorithm, up to the threshold term.

\section{simulations}\label{sec:simulations}
In this section, we present simulation examples illustrating some of the phenomena discussed above. 

Figure \ref{fig:wealth-dynamics} shows the step-by-step wealth shares over a duration of $10^6$ steps, as well as their moving averages over a window of $10$,$000$ steps for better visualization. After a brief initial phase, the no-regret learner (agent 2) converges to playing the correct distribution almost all the time, and it can be seen that its wealth approaches $1$.  

Figure \ref{fig:wealth-distribution-suffix} depicts the distribution of wealth shares for the agents at each time step during the second half of the simulation. It can be seen that the regret-minimizing agent 2 who is close to the correct distribution in the long run holds almost all the wealth almost all the time. Figure \ref{fig:wealth-distribution-prefix} depicts the distribution during the initial $10$,$000$ steps, showing a different picture. The Bayesian (agent 1) still managed to hold half of the wealth about half of this time, in line with the typical survival time estimate in Observation \ref{obs:survival-time-gainst-log-regret}. 
This illustrates the trade-off between short-term gains for the inaccurate Bayesian from converging quickly to a distribution close to the correct one, and the long-term dominance of the learner who converges to the correct distribution.

\begin{figure}[t!]
\vspace{-5pt}
    \centering
    \begin{subfigure}[b]{0.325\textwidth}
        \centering
        \includegraphics[width=1.16\textwidth]{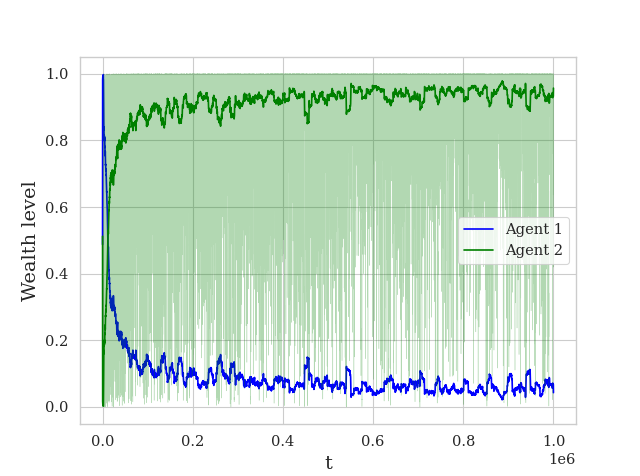}
        \vspace{-6pt}
        \caption{Wealth Dynamics. In green: player 2's wealth at each step. Solid lines show sliding window averages over the recent 10,000 steps for each $t$. }
        \label{fig:wealth-dynamics}
    \end{subfigure}
    \hfill
    \begin{subfigure}[b]{0.325\textwidth}
        \centering        \includegraphics[width=1.07\textwidth]{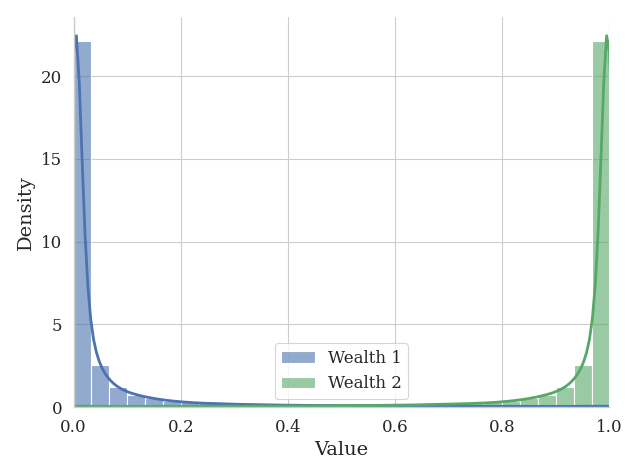}
        \vspace{-6pt}
        \caption{Wealth distribution in the second half of the $10^6$-step simulation. The distribution shown in blue is the wealth of player 1 and in green is of player 2.}
        \label{fig:wealth-distribution-suffix}
    \end{subfigure}
    \hfill
    \begin{subfigure}[b]{0.325\textwidth}
        \centering
        \includegraphics[width=1.07\textwidth]{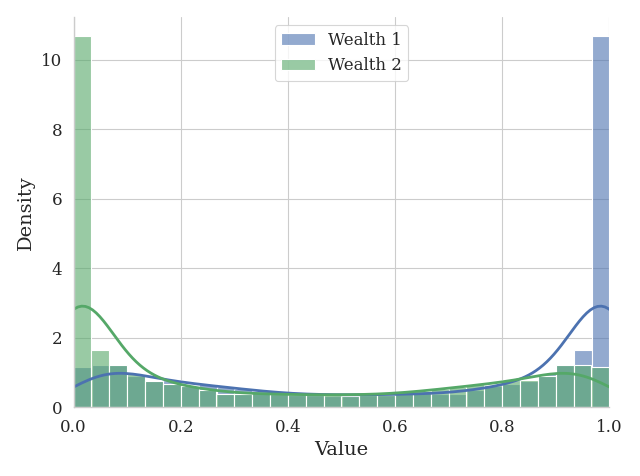}
        \vspace{-6pt}
        \caption{Wealth distribution in the 10,000 initial steps of the dynamic. The distribution shown in blue is the wealth of player 1 and in green is of player 2.}
        \label{fig:wealth-distribution-prefix}
    \end{subfigure}
    \caption{Example of wealth dynamics in competition between an inaccurate Bayesian learner with an error in its support (Agent 1) and a no-regret learner (Agent 2) who converges to the correct state distribution $q$.}
    \label{fig:wealth-figures}
\end{figure}
\noindent
\textbf{Imperfect Bayesian Learning:}
We begin by demonstrating how a Bayesian with a misspecified finite prior may initially do well, but eventually vanishes from the market.  
Consider the following simple example. There is a market with two assets where agent 1 is a Bayesian learner and agent 2 is a UCB learner, starting from equal wealth shares.\footnote{We use UCB for simplicity. The main factor here is the different regret rates and not the particular algorithm.} The distribution of states is $q = (0.7, 0.3)$, i.e., state 1 occurs with probability $0.7$, and state 2 otherwise. The Bayesian learner has a prior consisting of three hypotheses: $(0.8, 0.2)$, $(0.9, 0.1)$, $(0.3, 0.7)$, where the first is closest to $q$ in relative entropy, but still has an error. The UCB learner considers a similar set of distributions as its action set: $(0.7, 0.3)$, $(0.9, 0.1)$, $(0.3, 0.7)$, but the first of these being the correct distribution $q$. 

\begin{figure}[t!]
    \centering
    \vspace{-8pt}
    \begin{subfigure}[b]{0.325\textwidth}
        \centering
        \includegraphics[width=1.1\textwidth]{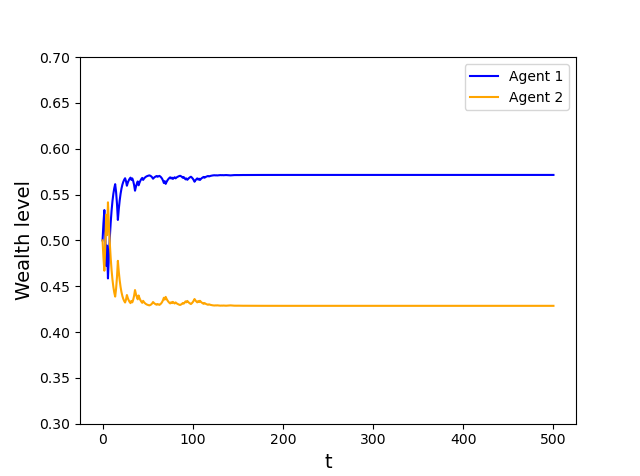}
        \vspace{-8pt}
        \caption{Two Bayesian learners.}
        \label{fig:wealth-dynamics-two-Bayesians}
    \end{subfigure}
    \hfill
    \begin{subfigure}[b]{0.325\textwidth}
        \centering
        \includegraphics[width=1.1\textwidth]{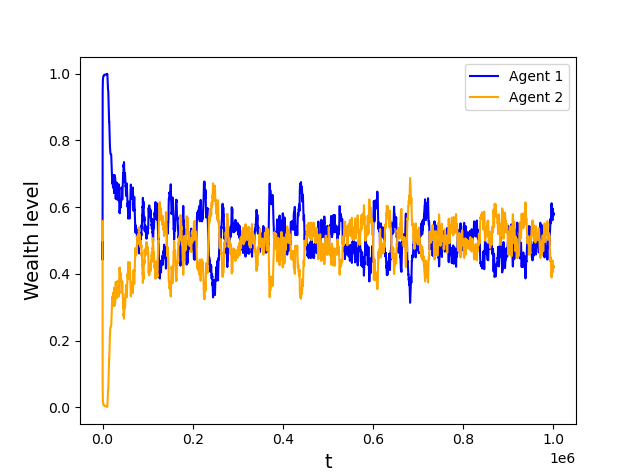}
        \vspace{-8pt}
        \caption{Two no-regret learners.}
        \label{fig:wealth-dynamics-two-UCB-agents}
    \end{subfigure}
    \hfill
    \begin{subfigure}[b]{0.325\textwidth}
        \centering
        \includegraphics[width=1.1\textwidth]{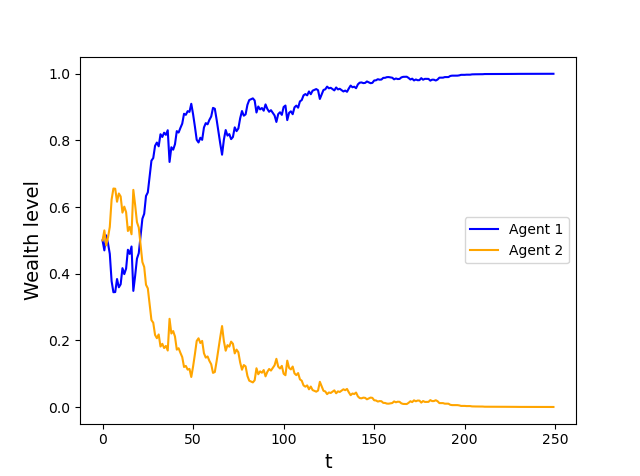}
        \vspace{-8pt}
        \caption{Bayesian vs. no-regret learner.}
        \label{fig:wealth-dynamics-Bayesian-vs-UCB}
    \end{subfigure}
    \caption{Wealth dynamics in a two-state two-player market. Figure \ref{fig:wealth-dynamics-two-Bayesians} shows the competition between two Bayesians where the agents have action sets of different sizes which include the correct model $q$. Agent $1$ considers three models and agent $2$ considers the same three models and one additional model. Figure \ref{fig:wealth-dynamics-two-UCB-agents} shows the same scenario for no-regret learners both using UCB. Figure \ref{fig:wealth-dynamics-Bayesian-vs-UCB} shows the competition between a Bayesian (agent 1) and a regret-minimizing UCB learner (agent 2) considering the same set of models.}
    \vspace{-5pt}
    \label{fig:wealth-figures-2}
\end{figure}

\bfpar{Dynamics with a Constant Regret Difference} Next, we consider an example of two Bayesian learners, both with priors including the correct distribution of states. By Theorem \ref{thm:regret-of-a-Bayesian}, their regret levels can differ only by an additive constant, as both are perfect Bayesians with constant regret. Consequently, these agents are bound to coexist in the market. The example demonstrates that while the quality of the priors---provided they include a correct model---does not affect survival, the initial learning phase, which depends on the prior, does influence the steady-state partition of wealth between the agents.  
In this example, as in the previous one, there are two states, with the underlying distribution given by $q = (0.7, 0.3)$. Agent 1 has a prior that includes the correct distribution and two other models: $(0.8,0.2), (0.7,0.3), (0.6,0.4)$. Agent 2's prior includes one additional model: $(0.8,0.2), (0.7,0.3), (0.6,0.4), (0.5, 0.5)$. Both agents start with a uniform prior over the options they consider. Thus, agent 2 places smaller weight initially on the correct model than does agent 1. Figure \ref{fig:wealth-dynamics-two-Bayesians} depicts the wealth dynamics of these learners, starting from equal wealth levels. 

Figure \ref{fig:wealth-dynamics-two-UCB-agents} shows the same scenario in competition between two no-regret learners using UCB. While the worst-case bound of UCB suggests that the learner with the larger action space could suffer higher regret (by a constant factor), we see empirically that under a steady stochastic process, the agents are comparable in the long run, where only in the initial learning phase agent 1 with the smaller action set has an advantage. In comparison to the Bayesian learners in Figure \ref{fig:wealth-dynamics-two-Bayesians}, no-regret dynamics continue to be randomized and ``forget'' the gap in wealth from the initial learning phase. 

\bfpar{Comparison of Bayesian vs. No-Regret Learners} Finally, in Figure \ref{fig:wealth-dynamics-Bayesian-vs-UCB} we present an example of the wealth dynamics of a Bayesian learner (agent 1) in competition with a UCB no-regret learner (agent 2). The system is similar to the ones described above, with two states and $q = (0.7,0.3)$. Both agents consider the models $(0.8,0.2), (0.7,0.3), (0.6,0.4)$, so for both,  the correct model $q$ is in the support. As expected according to Theorem \ref{thm:regret-of-a-Bayesian}, the perfect Bayesian who has constant regret drives out of the market the no-regret learner whose regret is increasing in time.

\section{Discussion}\label{sec:discussion}
Our analysis uses a stylized  market environment to 
study competition between heterogeneous learning agents in asset markets. The analysis isolates a force that is central in competitive investment settings: under multiplicative wealth dynamics, differences in learning speed translate into differences in long-run survival.  
In particular, in such environments, relative wealth evolves according to cumulative differences in growth rates, and the survival criterion formalizes 
how sustained differences in growth lead to vanishing or dominant wealth shares over time.  

Following the prior literature, we study learners who optimize log-of-wealth objectives. 
As discussed in Section \ref{sec:model}, there is a longstanding debate about the centrality of log utility in work on financial markets and whether it is normatively justified. From the perspective of online learning, the log-of-wealth objective has the technical advantage that it makes growth rates additive over time, and so allows the use of no-regret learning methods in such compounding payoff settings.
  
The Arrow-securities model serves here as an analytical device for studying this force. The complete-market structure provides a clean representation of state-contingent investment and makes transparent the relation between learning, regret, and wealth accumulation. For further exposition on complete markets and the Arrow model, see Appendix~\ref{sec:appendix-Arrow-securities}. The broader point, however, is not tied specifically to the Arrow formulation. What is essential is that agents repeatedly allocate wealth across risky opportunities, and that realized payoffs feed back into future investment capacity. The Arrow-securities representation 
should be understood as a convenient basis for the analysis:  
investors may act in the natural space of investment opportunities and need not formulate their decisions in Arrow-security coordinates.

An additional modeling choice is to restrict attention to full-support investment profiles. In multiplicative environments, allocating all investment to a strict subset of states may lead to irreversible loss. Any reasonable long-run learning rule must therefore avoid investments that risk bankruptcy in a single step. From this perspective, full support is not just a technical assumption, but a basic robustness requirement under uncertainty.

A central theme of the paper is the role of Bayesian learning as a benchmark of model-based learners. On the one hand, our results align with the economics literature on market selection: a perfect Bayesian with positive prior on the true model survives even against optimal competitors and drives out slower learners. We provide a precise characterization of this statement in the language of regret: survival is equivalent to maintaining a bounded regret gap relative to all  competitors. 

On the other hand, our analysis shows a complementary lesson. The power of Bayesian learning rests on strong informational assumptions. If the true model lies even slightly outside the prior's support, or if the update rule is implemented with small errors, regret becomes linear, which leads to long-run failure in competition. Thus, while Bayesian learning is optimal under correct specification, approximate Bayesian learning---which is closer to what may be implementable in practice---does not give approximately optimal performance in this environment; instead, it may generate short-term gains, but leads to vanishing wealth in the long run. In contrast, payoff-based no-regret methods require far less knowledge of the data-generating process and guarantee sublinear regret without specifying a model class. When accurate structural knowledge is unavailable, such robustness becomes a decisive advantage. Section~\ref{sec:robust-Bayes-update} builds on this insight and takes a step toward bridging these approaches.

Our emphasis in this study is on long-run survival, but the connection between regret and wealth is not merely asymptotic. At any horizon, differences in regret translate directly into differences in wealth shares. In particular, this connection shows how regret rates translate into vanishing rates: for instance, a learner incurring logarithmic regret against a constant-regret competitor sees their relative wealth vanish at a polynomial rate, whereas a linear regret gap leads to exponentially vanishing wealth. This allows one to evaluate transient survival times, as illustrated in Section~\ref{sec:noisy-priors-Bayesian}. Because the relation holds for any investment strategy, regret rates provide a meaningful performance metric in competitive markets. Even when learners are not explicitly minimizing regret, estimating their regret rates, whether analytically or empirically, offers insight into their competitive fitness.

Finally, our model focuses on investment markets, where payoffs compound multiplicatively. In other competitive environments, such as auctions or pricing in retail platforms where learning agents are broadly deployed, payoffs are often modeled with additive per-period payoffs (with or without budget constraints). In such settings, slower learners may still earn positive payoffs each period without facing extinction. However, if feedback effects link current performance to future capacity---for example, through feedback between past performance and future budget constraints, or through reinvestment of profits---relationships analogous to Equation~(\ref{eq:logshare-as-regret}) may arise, tying speed of learning to long-run dominance. More broadly, a growing literature shows that the choice and configuration of learning algorithms in markets can themselves generate strategic effects that alter market outcomes, including algorithmic collusion in pricing and inefficient equilibria in auctions (see references in Section~\ref{sec:related-work}). Our results suggest that in environments with multiplicative wealth dynamics, the choice of learning method may have especially strong long-run consequences.


\section{Conclusion}\label{sec:Concln}
As markets become increasingly populated by learning agents, understanding which learning approaches succeed in competitive environments becomes a foundational question. Our work brings together perspectives from economics and computer science by analyzing the survival and performance of Bayesian and no-regret learners in complete asset markets. We connect the notions of regret and survival in a competitive market, and show that regret minimization, while robust and widely applicable, does not guarantee survival: a Bayesian investor with correct support and proper updates can drive out no-regret learners, even if their regret is vanishing. Conversely, Bayesian learning is fragile---slight misspecification in priors or slight mistakes in belief updates may lead to linear regret and eventual extinction from the market.

These insights motivate the study of more robust ways to utilize Bayesian learning and combine advantages from both approaches. 
We proposed and analyzed two simple variants of Bayesian learning: one uses regularization to deal with distribution shifts, and the other uses a hybrid approach that follows Bayesian updates while tracking regret relative to a no-regret benchmark. The resulting methods preserve the fast convergence of Bayesian learning in well-specified settings while improving robustness when the environment changes or the prior is misspecified, offering simple and practical steps toward a best-of-both-worlds approach. 
Our findings highlight the speed of learning as a crucial determinant of survival and show that even constant factors in regret rates matter in competitive markets. We view the development of robust 
Bayesian learning methods, guided by principles from online learning theory, as a promising direction for future work. 

Another interesting direction for future work is to consider the impact of competing learners on aggregate outcomes such as equilibrium  prices. Both long-run security prices, determined by who survives, and intermediate term security prices, determined by the evolution of investment rules and wealth shares, could be analyzed. Our focus has been on the survival of individual agents and they determine long run prices. But intermediate term prices also depend on how agents respond to the feedback they receive over time from markets and how wealth shares evolve between these agents. Both Bayesian and standard no-regret learning cause agents to invest more in assets that pay off and wealth share increases for agents who invest relatively more in these assets. So intermediate term security prices respond positively to past payoffs. These learning and wealth share evolution effects induce positive serial correlation---or momentum---in security prices.

\section*{Acknowledgments} 
We thank the anonymous referees at EC 2025 for useful feedback on an earlier version of the paper. 
\'Eva Tardos was supported in part by 
AFOSR grant FA9550-23-1-0410,  AFOSR grant FA9550-231-0068, and ONR MURI grant N000142412742.

\appendix
\section*{APPENDIX}
\section{Proofs}\label{sec:appendix-proofs}
We provide in the following the proofs for the theorems stated in the main text.

\vspace{5pt}
\begin{proof} (Theorem \ref{thm:survival-by-finite-regret-gap})
 This result follows from the analysis leading to Equation (\ref{eq:logshare-as-regret}). For the first part of the theorem, we have $R^n(T) - R^m(T) = -\log(r^{nm}) - \log(r_0^{nm}) = \log(r^{mn}) + const$. If the left-hand side diverges to $+\infty$, then $r^{mn} = w^m/w^n$ diverges as well, which is only possible if $w^n \rightarrow 0$. For the second part, clearly this is a necessary condition due to the first part of the theorem: if an agent vanishes, it does not survive. To see that this condition is sufficient, if there exists a constant $c > 0$ such that $R^n(T) - R^m(T) < c$ at all times and for all $m \neq n$, then the wealth ratio $r^{mn}$ with any other player is bounded, and thus $w^n$ is bounded from below. 
\end{proof}
\vspace{5pt}

\begin{proof} (Theorem \ref{thm:regret-of-q}) Consider an agent indexed by $q$ using the state distribution $q$ as its investment strategy; i.e., $\alpha^q_{t} = q$ for all $t$, and we denote this constant sequence by $q_{1:T}$. Agent $q$ competes with a magic agent. The regret of the agent playing strategy $q$ is denoted $R^T(q_{1:T})$, and for each $T$ denote its expectation by $R = \E[R^T(q_{1:T})]$. On the one hand, using the fact that $\hat{q}$ has zero regret by definition, 
from Equation \ref{eq:logshare-as-regret} we have $$ R^T(q_{1:T}) = \log(r^{\hat{q}q}_T) - \log(r^{\hat{q}q}_0). $$ 
On the other hand, taking expectation of Equation \ref{eq:logshares}, we have 
\begin{align*} 
\E \big[ & \log(r^{\hat{q}q}_T)\big]\\ 
& = \E \big[\sum_s \sum_t \mathbf{1}_{s_t = s} \log(\frac{\alpha^{\hat{q}}_{s}}{q_{s}})\big]+ \log(r^{\hat{q}q}_{0}) \\ 
& = T \E \big[\sum_s \frac{n^s_T}{T} \log(\frac{n^s_T/T}{q_s})\big]+ \log(r^{\hat{q}q}_{0}) \\ 
& = T \cdot \E[I_{\hat{q}}(q)] + \log(r^{\hat{q}q}_{0}),  
\end{align*} 
where $n^s_T$ is the empirical count of state $s$ until time $T$. 
And thus, we get 
$$ 
R = T \cdot \E \big[I_{\hat{q}}(q)\big]. 
$$ 
That is, the expected regret equals the relative entropy (a.k.a. KL-divergence) between an empirical distribution and the state distribution scaled by $T$. 
Here, we use a result from information theory, showing that the expected KL-divergence of the empirical distribution decreases as the alphabet size divided by the sample size\footnote{We are interested in the expectation. \cite{mardia2020concentration} also provide concentration bounds for this result.} (see, e.g., Lemma 13.2 in \cite{polyanskiy2024information}). 
In our case, 
$\E[I_{\hat{q}}(q)] \equiv \E[D_{KL}(\hat{q}||q)] = \frac{S-1}{2T} + o(\frac{1}{T})$. 
Thus, 
$ R = \frac{S-1}{2} + o(1) $, 
and the expected regret is bounded by a constant depending only on the number of states. 
\end{proof}
\vspace{5pt}

\begin{proof} (Theorem \ref{thm:regret-of-a-Bayesian}) 
Starting from the second point, we establish a  pathwise bound. The realized wealth of the Bayesian is a weighted average of the wealths of the models in its prior: $w_T^B = \sum_k \lambda_0^k w_T^{\theta^k}$. Because the true distribution $q$ is in the prior, we have on every single sample path that $w_T^B \ge \lambda_0^q w_T^q$. Using Equation~\ref{eq:logshare-as-regret} that connects wealth ratios to regret differences, we obtain:
\begin{equation*}
R_T^B - R_T^q = 
\log\Big(\frac{w_T^q}{w_T^B}\Big) 
- 
\log\Big(\frac{w_0^q}{w_0^B}\Big)
\leq 
\log\Big(\frac{1}{\lambda_0^q}\Big)
- 
\log\Big(\frac{w_0^q}{w_0^B}\Big).
\end{equation*}

For the first point, taking the expectation of the pathwise bound yields $\mathbb{E}[R_T^B] \le \mathbb{E}[R_T^q] + \log(1/\lambda_0^q)$. By Theorem \ref{thm:regret-of-q}, the expected regret of playing $q$ is bounded by a constant. Therefore, the Bayesian's expected regret is also bounded by a constant at all times.
\end{proof}
\vspace{5pt}

\begin{proof} (Theorem \ref{thm:wrong-Bayesians-vanish}) 
Let $n$ be the index of a Bayesian agent with $q$ not within its support, and let $q' \neq q$ be the strategy within the support that is closest to $q$ in relative entropy. By Proposition \ref{thm:bayesian-convergence}, the strategy used by the Bayesian agent, $\alpha^n_t$, converges almost surely to $q'$. Hence, intuitively, for sufficiently large $t$, the strategies remain bounded away from $q$, leading to linear regret from that point onward. 
In more detail, let $n$ be the index of a Bayesian agent with $q$ not within its support, and let $q' \neq q$ be the strategy within the support closest to $q$ in relative entropy. By Proposition 1, the agent's strategy $\alpha_t^n$ converges almost surely to $q'$.
The utility difference at time $t$ relative to strategy $q$ is $U_t(q) - U_t(\alpha_t^n)$. Since $\alpha_t^n \to q'$ almost surely and the states $s_t$ are drawn i.i.d.\ from $q$, the Strong Law of Large Numbers yields:

$$\lim_{T \to \infty} \frac{1}{T} \sum_{t=1}^T \left(U_t(q) - U_t(\alpha_t^n)\right) = I_q(q') \quad a.s.$$

Because $q' \neq q$, we have $I_q(q') > 0$. Thus, agent $n$'s regret relative to $q$ grows linearly almost surely. Because $\hat{q} \to q$ almost surely, the realized pathwise regret of strategy $q$ is sublinear ($T \cdot I_{\hat{q}}(q) = o(T)$). Therefore, the Bayesian's total realized regret $R^n(T)$ grows linearly in $T$ almost surely.
Now consider a competing no-regret learner $m$ with sublinear regret $R^m(T)$. Then, the regret difference $R^n(T) - R^m(T) \to \infty$ almost surely. By Theorem 1, the imperfect Bayesian agent vanishes from the market.
\end{proof}
\vspace{5pt}

\begin{proof} (Theorem \ref{thm:wrong-update-Bayesians-vanish}) 
Let $\epsilon > 0$. 
The log-likelihood ratio under the noisy update rule takes the following form:
\begin{align}\label{eq:inacurate-update-log-ratio}
    \log\Big(
    \frac{P_t(\theta_a)}{P_t(\theta_b)}
    \Big) 
    &=  
    (1 + \lmbt) L(s_t) 
    \nonumber \\&+ 
    (1 - \lmbt) \log \Big(\frac{P_{t-1}(\theta_a)}{P_{t-1}(\theta_b)}\Big) \nonumber \\
    &= 
    (1 + \lmbt)\sum_{\tau=0}^{t-1} L(s_{t - \tau}) \prod_{k=0}^\tau (1 - \eta_k)\\ 
    \nonumber 
    &+ 
    \log\Big(
    \frac{P_0(\theta_a)}{P_0(\theta_b)}
    \Big)
    \prod_{\tau=0}^t
    (1 - \eta_{\tau}),
\end{align}
where the empty product equals one, and we define $\eta_0 = 0$.
The products can be simplified:
\begin{align*}
    \prod_{\tau=0}^t
    (1 - \eta_\tau)  = 
    (1 - \lmb)^{n_+(t)}
    (1 + \lmb)^{n_-(t)},
\end{align*}
where $n_+(t)$ is a binomial random variable counting the number of times $\eta_{\tau \leq t} = \lmb$, and $n_-(t) = t - n_+(t)$. 
The last term can be written as  
\begin{align*}
    &
    \log\Big(
    \frac{P_0(\theta_a)}{P_0(\theta_b)}
    \Big)
    \cdot 
    (1 - \lmb)^{n_+(t)}
    (1 + \lmb)^{n_-(t)} \\
    &= 
    \log\Big(
    \frac{P_0(\theta_a)}{P_0(\theta_b)}
    \Big)
    \cdot 
    \Big[
    \Big(
    \frac{1 - \lmb}
    {1 + \lmb}\Big)^{n_+(t)}
    \Big] \cdot
    (1 + \lmb)^t.
\end{align*}

Intuitively, since $\E[n_+] = T/2$, this should be close to $(1 - \lmb^2)^{\nicefrac{t}{2}}$ with high probability as $t \rightarrow \infty$, which converges to zero for $\eta < 1$. Formally, we have the following claim.
    \begin{claim}
        Let $\epsilon > 0$.
        $
        \lim_{t \rightarrow \infty} \Pr \Big[ \left(\frac{1 - \lmb}{1 + \lmb}\right)^{n_+(t)} \cdot (1 + \lmb)^t > \epsilon \Big] = 0.
        $
    \end{claim}
    
    \begin{proof}
        Let $\delta > 0$. We need to show that there exists $T$ such that for all $t > T$, the event
        \[
        \Big(\frac{1 - \lmb}{1 + \lmb}\Big)^{n_+(t)} \cdot (1 + \lmb)^t > \epsilon
        \]
        has probability less than $\delta$. 
        Dividing by $(1 + \lmb)^t$ and taking the logarithm, we get 
        $
        n_+(t) \cdot \big(\log(1 - \lmb) - \log(1 + \lmb)\big) > \log(\epsilon) - t \log(1 + \lmb).
        $.
        Rearranging, we have
        \[
        n_+(t) < \frac{t \log(1 + \lmb) - \log(\epsilon)}{\log(1 + \lmb) - \log(1 - \lmb)}.
        \]
        Denote $c_{\epsilon} = \frac{\log(\epsilon)}{\log(1 + \lmb) - \log(1 - \lmb)}$,
        so
        $
        n_+(t) < \frac{\log(1 + \lmb)}{\log(1 + \lmb) - \log(1 - \lmb)} \cdot t - c_{\epsilon}.
        $ 
        Next, we observe that the coefficient of $t$ is strictly less than $1/2$. To see this, define
        \begin{align*}
        c &= \frac{1}{2} - \frac{\log(1 + \lmb)}{\log(1 + \lmb) - \log(1 - \lmb)} \\
        &= 
        -\frac{\log(1 + \lmb) + \log(1 - \lmb)}{2 \big(\log(1 + \lmb) - \log(1 - \lmb)\big)}\\ 
        &= 
        \frac{-\log(1 - \lmb^2)}{\text{Positive Number}} > 0.
        \end{align*}
        Thus, we obtain that the event we wish to bound is
        $
        n_+(t) < \Big(\frac{1}{2} - c\Big)t - c_{\epsilon}.
        $
        Using Hoeffding's inequality, using that $\E[n_+] = \nicefrac{t}{2}$, we bound the probability of this event:
        \[
        \Pr\Big[\big|n_+(t) - \nicefrac{t}{2}\big| 
        > 
        ct + c_{\epsilon}\Big] 
        < 
        2e^{-\frac{2(ct + c_{\epsilon})^2}{t}} < 
        2e^{-\frac{1}{2}ct}.
        \]
        Setting $T = \lceil \frac{2}{c} \ln\left(\frac{2}{\delta}\right)
        \rceil$, we conclude that for all $t > T$, the probability of the event
        \[
        (1 + \lmb)^t \cdot \Big(\frac{1 - \lmb}{1 + \lmb}\Big)^{n_+(t)} > \epsilon
        \]
        is less than $\delta$, as required. Hence, the product approaches zero for large $t$:
        $
        \lim_{t \rightarrow \infty} \prod_{\tau=1}^t (1 - \eta_\tau) = 0
        $ almost surely, proving the claim.
    \end{proof}

    \vspace{5pt}
    \noindent
    {\em Remark:} It is worth noting here that the almost-sure behavior of the product is very different from its expectation, which equals one: 
    $$\sum_{k=0}^T \binom{T}{k} p^k(1 - \lmb)^k \cdot (1 - p)^{T - k} (1 + \lmb)^{T - k} = \big(p (1 - \lmb) + (1 - p) (1 + \lmb)\big)^T = 1 \ \text{ (since $p = 1/2$)}.$$

    The above claim implies that the second term in Equation (\ref{eq:inacurate-update-log-ratio}) converges to zero almost surely. In other words, the learner forgets the prior, much like a perfect Bayesian learner (albeit at a somewhat faster rate in this case). For the first term of Equation (\ref{eq:inacurate-update-log-ratio})
    $$(1 + \lmbt)\sum_{\tau=0}^{t-1} L(s_{t - \tau}) \prod_{k=0}^\tau (1 - \eta_k),$$
    a similar argument applies: 
    with high probability the product for large values of $\tau$ contributes negligibly to the sum, but there remains a random contribution from the smaller $\tau$ terms.  
    To see why there is no convergence, observe that the coefficient $(1 + \lmbt)$ oscillates randomly around $1$. More formally, let 
    $Y_\tau = \prod_{k=0}^\tau (1 - \eta_k)$, noting that $Y_0 = 1$, and let $Z_t  = \sum_{\tau=0}^{t-1}L(s_{t-\tau})Y_\tau$. Expanding the sum, we have: $Z_t = L(s_t)Y_0 + \dots + L(s_1)Y_{t-1}$ and for the next time step $Z_{t+1} = L(s_t)Y_0 + \dots + L(s_1)Y_{t}$. Since $Y_t = (1-\lmbt)Y_{t-1}$, this simplifies to $Z_{t+1} = (1 - \lmbt)Z_t + L(s_{t+1})$. 
    Clearly, this sequence does not converge, as both $L(s_{t+1})$ and $\lmbt$ take independent random values at each step.
\end{proof}
\vspace{5pt}

\begin{proof} (Lemma \ref{thm:expected-regret-lemma}) 
    The result follows from the definition of regret as a utility difference, as discussed in Section \ref{sec:compare}. Specifically, the regret of strategy sequence $\alpha^n_{1:T}$ is $R^T(\alpha^n_{1:T}) = U(\alpha^{\hat{q}}_{1:T}) - U(\alpha^n_{1:T})$, also, $R = \E[U(\alpha^{\hat{q}}_{1:T}) - U(q_{1:T})]$, and denote the regret of $\alpha^n_{1:T}$ with respect to the benchmark $q_{1:T}$ as $R^\alpha_q = U(q_{1:T}) - U(\alpha^n_{1:T})$. We have 
    $
    \E[R^T(\alpha^n_{1:T})] = \E[R^\alpha_q] + R = \E\Big[ T \sum_s \frac{n^s_T}{T} \log(\frac{q_s}{\alpha^n_{s}}) \Big] + R = T \cdot I_q(\alpha^n) + R
    $.   
\end{proof}
\vspace{5pt}

\begin{proof} (Proposition~\ref{thm:dist-shift-linear-regret-proposition})
    Consider a market process with two states and $T$ time steps. Let $1/2 < p < 1$, and define $q_1 = (p, 1-p)$, $q_2 = (1-p, p)$, and $Q = \{q_1, q_2\}$. Suppose that for the first $T_1$ steps the data-generating distribution is $q_1$, and for the remaining $T_2 = T - T_1$ steps it is $q_2$. At any time $t \leq T_1$, taking expectation in Equation~(\ref{eq:log-odds-by-n}), we have for a Bayesian learner starting from a uniform prior:
    $$
    \E[\log \frac{\lambda^1_t}{\lambda^2_t}] = \E[n^1_t] \cdot \log \frac{p}{1-p} + \E[n^2_t] \cdot \log \frac{1-p}{p},
    $$
    where, recall, $n^s_t$ denotes the number of times state $s$ has occurred up to time $t$.
    At time $T_1$ we get:
    $$
    \E[\log \frac{\lambda^1_{T_1}}{\lambda^2_{T_1}}] = 
    pT_1 \cdot \log \frac{p}{1-p} + (1-p)T_1 \cdot \log \frac{1-p}{p}.
    $$

    After the shift, at $t = T_1 + \tau$, using the same formula and the learner's state at time $T_1$, we have:
    $$
    \E[\log \frac{\lambda^1_t}{\lambda^2_t}] 
    = 
    \Big(
    pT_1 \cdot \log \frac{p}{1-p} + (1-p)T_1 \cdot \log \frac{1-p}{p} 
    \Big)
    +
    \Big(
    (1-p)\tau \cdot \log \frac{1-p}{p} + p\tau \cdot \log \frac{p}{1-p}
    \Big).
    $$
    Grouping terms, we get:
    $$
    \E[\log \frac{\lambda^1_t}{\lambda^2_t}] 
    = 
    p(T_1+\tau) \cdot \log \frac{p}{1-p} + (1-p)(T_1+\tau) \cdot \log \frac{1-p}{p}.
    $$

    Intuitively, for any $\tau < T_1$, we have $\E[n^1_{T_1 + \tau}] > \E[n^2_{T_1 + \tau}]$, and so $\lambda^1_t > \lambda^2_t$ during steps $t = T_1,\dots,T_1 + \tau$. This implies that during the second phase, the learner still assigns probability greater than $1/2$ to state~1, despite state~1 occurring with probability less than $1/2$. This discrepancy leads to regret proportional to $\tau$.     Concretely, take $T_1 = 2T/3$, $p = 3/4$, and $\tau = T/3$. Then during steps $t = T_1 + 1,\dots,T$ we have:
    $$
    \E[\log \frac{\lambda^1_t}{\lambda^2_t}] 
    > 
    \E[\log \frac{\lambda^1_T}{\lambda^2_T}] 
    = 
    \frac{3T}{4} \cdot \log 3 
    -
    \frac{T}{4} \cdot \log 3
    =
    \frac{\log 3}{2} T > \frac{1}{2} T.
    $$
    So the empirical probability the learner assigns to state 1 is at least $1/2$, whereas the true probability during the last $T/3$ steps is $1/4$. Therefore, similarly to the proof of Theorem \ref{thm:wrong-Bayesians-vanish}, 
    since there is a gap between the assigned and true probabilities, the regret accumulated during these steps after $2T/3$ is linear; in expectation it is at least the regret of the strategy $(\tfrac{1}{2}, \tfrac{1}{2})$, which is
    \[
    \frac{T}{3}\cdot I_{q_2}\left(\tfrac{1}{2}, \tfrac{1}{2}\right) = \frac{3 \log 3 + 4 \log 2}{4} \cdot \frac{T}{3} > \frac{T}{23}.
    \]
    During the first stage, until $2T/3$, the learner incurs a constant regret, so the total regret is linear in $T$.
\end{proof}
\vspace{5pt}

\begin{proof} (Theorem~\ref{thm:robust-Bayesian}) Let us consider a phase of the process in which the data generating distribution is $q$. For notational convenience, we overload the symbol $q$ to also denote the index of this distribution, so that $\lambda^q$ denotes the weight the learner assigns to this correct model.  
By the update rule~(\ref{eq:robust-Bayes-update}), for any $k \neq q$:
$$
\log \frac{\lambda^k_{t+1}}{\lambda^q_{t+1}} 
= 
\log \frac{\tilde{\lambda}^k_{t+1} + \epsilon_t}{\tilde{\lambda}^q_{t+1} + \epsilon_t}
=
\log \frac{\tilde{\lambda}^k_{t+1}}{\tilde{\lambda}^q_{t+1}} 
+ 
\log \big(1 + \tfrac{\epsilon_t}{\tilde{\lambda}^k_{t+1}} \big) 
- 
\log \big(1 + \tfrac{\epsilon_t}{\tilde{\lambda}^q_{t+1}} \big).
$$
Since for any $x > 0$ it holds that $0 < \log(1 + x) \leq x$, we have
$$
\log \frac{\lambda^k_{t+1}}{\lambda^q_{t+1}} 
\leq
\log \frac{\tilde{\lambda}^k_{t+1}}{\tilde{\lambda}^q_{t+1}} 
+
\frac{\epsilon_t}{\tilde{\lambda}^k_{t+1}}.
$$
Applying the Bayesian update rule for the log-odds and taking expectations, we get
$$
\E\Big[\log \frac{\lambda^k_{t+1}}{\lambda^q_{t+1}} \Big]
\leq
\log \frac{\lambda^k_t}{\lambda^q_t}
-
I_q(\theta^k)
+
\E\Big[ \frac{\epsilon_t}{\tilde{\lambda}^k_{t+1}} \Big].
$$
Now consider two cases. First, suppose that $\tilde{\lambda}^k_{t+1} < t^{-\tfrac{3}{2}}$. For $t$ large enough (e.g., $t > 6$), we have
$$
\lambda^k_{t+1} < \tilde{\lambda}^k_{t+1} + \epsilon_t \le t^{-\tfrac{3}{2}} + t^{-2} < t^{-\tfrac{4}{3}}.
$$
Intuitively, once $\tilde{\lambda}^k_{t+1}$ becomes small, it stays small.
In the second case, $\tilde{\lambda}^k_{t+1} \geq t^{-\tfrac{3}{2}}$. Then, we have
\begin{equation}\label{eq:robust-Bayes-update-convergence}
\E\Big[\log \frac{\lambda^k_{t+1}}{\lambda^q_{t+1}} \Big]
\leq
\log \frac{\lambda^k_t}{\lambda^q_t}
-
I_q(\theta^k)
+
O\big(t^{-\frac{1}{2}}\big).
\end{equation}
Since $I_q(\theta^k)$ is a positive constant for any $\theta^k \neq q$, we see that for large enough $t$, the log-odds of $k \neq q$ relative to $q$ decrease at a linear rate, implying that the posterior distribution approaches exponentially fast the true model $q$. Thus, after a number of steps depending logarithmically on the value of  $\lambda^q$ at the start of the phase,  
we reach---with high probability over the randomness in the data-generating process---the first case where $\tilde{\lambda}^k_{t+1} < t^{-\tfrac{3}{2}}$. 

\vspace{5pt}
\noindent
\textbf{Constant regret for a stationary process:} 
Given that $\lambda_0^q$ is a constant, after some constant time $T_0$ (independent of $T$), it holds with high probability that for all $t > T_0$ and all $k \neq q$, we have $\lambda^k_t < t^{-\tfrac{4}{3}}$.
Using this, we can bound the regret of player $n$ under our strategy. Let $I$ denote the relative entropy of the worst model in the prior: $I = \max_k I_q(\theta^k)$. With high probability,
$$
R_T(\alpha_{1:T}^n) - R \leq 2\sum_{t=1}^{T} I_q(\alpha^n_t) \leq 2\sum_{t=1}^{T} I \cdot t^{-4/3} < 2\sum_{t=1}^{\infty} I \cdot t^{-4/3} = 2I \cdot \zeta(4/3) \approx 2I \cdot 3.601.
$$
Where $\zeta(x) = \sum_{n=1}^\infty \tfrac{1}{n^x}$ denotes the Riemann zeta function which has finite value for $x>1$.

\vspace{5pt}
\noindent
\textbf{Logarithmic regret after a distribution shift:} The above analysis shows one side of the regularization method: when the weights on models $k \neq q$ are large, they decay exponentially fast, similar to the standard Bayesian update. 

The other side is what happens after a distribution shift. The key point is that while the weights on incorrect models become small enough to ensure constant regret, the regularization prevents them from becoming exponentially small. Suppose that after some history of length $T_1$, the stochastic process shifts to a new distribution, which we denote by $q'$. 
A lower bound on $\lambda^{q'}_{T_1+1}$ is the amount of regularization added to it in the previous step, namely,
$$
\lambda^{q'}_{T_1+1} \geq \frac{\epsilon_{T_1}}{1 + K \epsilon_{T_1}} \geq \frac{1}{K {T_1}^2}.
$$
The last inequality holds since the shift occurs at $t > 1$.

So at time $T_1 + 1$, we can think about the situation as a new learning process with a prior weight of at least $\lambda^{q'}_{T_1+1} \geq \frac{1}{K {T_1}^2}$. 
Because of exponential convergence, the time to ``forget'' this prior is logarithmic in $T_1$. 
Specifically, the exponent for model $k$ is $I_{q'}(\theta^k)$. 
Let $I$ now be the minimal entropy over models $k \neq q'$. The expected time $\tau$ it takes to reach $\lambda_t^k\le \lambda_t^{q'}$ for all $k \neq q'$ and eliminate this prior is $\tau = \tfrac{2}{I} \log T_1 + \text{constant}$. 
After this time, we already know that the regret incurred until the next distribution shift, or reaching the time horizon $T$, is bounded by a constant independent of $T$. 
\end{proof}
\vspace{5pt}

\begin{proof} (Theorem~\ref{thm:hybrid-Bayes-No-regret}) 
Let us start from point (2) in the theorem. Let
$$
\tau=\inf\{t\geq 1 : D_t\geq b_t\},
$$
with $\tau=\infty$ if no switch occurs, and let $\alpha^{\mathrm{sw}}_{1:T}$ denote the played sequence.
We first compare the switching strategy to the no-regret strategy pathwise. Consider two cases:

\noindent
Case 1: If $\tau > T$, then the algorithm plays Bayes throughout rounds $1,\dots,T$, and therefore
$$
R_T(\alpha^{\mathrm{sw}}_{1:T}) - R_T(\alpha^N_{1:T})
=
D_T
<
b_T. 
$$

\noindent
Case 2: If $\tau \leq T$, then the switching strategy agrees with Bayes up to time $\tau$ and with the no-regret strategy in all remaining steps. Hence
$$
R_T(\alpha^{\mathrm{sw}}_{1:T}) - R_T(\alpha^N_{1:T})
=
D_\tau.
$$
Since $\tau$ is the first crossing time, we have $D_{\tau-1} < b_{\tau-1}$. Since $b_t$ is increasing and $|d_t|\leq L$,
$$
D_\tau = D_{\tau-1} + d_{\tau} \leq b_{\tau-1} + L \leq b_T + L.
$$
Thus, for every realization,
$$
R_T(\alpha^{\mathrm{sw}}_{1:T})
\leq
R_T(\alpha^N_{1:T}) + b_T + L.
$$
Now, assuming any $O(\sqrt{T})$ guarantee for the no-regret algorithm generating $\alpha^N$ (or, in particular, logarithmic-regret portfolio-optimization algorithms such as those discussed in Section~\ref{sec:compare}) concludes the proof of the second part of the theorem. 

For the first part, to show that switching is rare under correct specification, we can look at the process $D_t$ under the prior mixture. This mixture is the distribution obtained by first drawing a model from the Bayesian prior and then generating the data from that model. 

Let $P_{\theta^k}$ denote the probability  over state sequences when $\theta^k$ is the true distribution, and let
$$
P_{\mathrm{mix}}=\sum_{k=1}^K \lambda_0^k P_{\theta^k}
$$
be the mixture distribution induced by the Bayesian prior. 
Note that $P_{\mathrm{mix}}$ is not the true data-generating distribution (under any fixed model $\theta^k$). It is only a proof device constructed for bounding the probability of a false switch in the case when the true distribution is in the support of the Bayesian prior.   

Let $\mathcal H_{t-1}$ be the history up to time $t - 1$, 
including the realization of any internal randomness of the no-regret algorithm. 
Under the prior, model $\theta^k$ is true with probability $\lambda_0^k$. After observing the history up to time $t - 1$, 
the posterior probability of model $\theta^k$ is $\lambda_{t-1}^k$. 
Since under model $\theta^k$, the next state is distributed according to $\theta^k$, 
the conditional distribution of the next state under the mix is the weighted average of these models' distributions with the posterior weights at $t - 1$:   

$$
P_{\mathrm{mix}}(s_t=s | \mathcal H_{t-1})
=   
\sum_{k=1}^K \lambda_{t-1}^k \theta_s^k.
$$
The Bayesian action is exactly 
$$
\alpha_t^B=\sum_{k=1}^K \lambda_{t-1}^k \theta^k,
$$
which maximizes conditional expected log growth under this distribution, so 
$$
\E_{P_{\mathrm{mix}}}[d_t | \mathcal H_{t-1}] \leq 0.
$$
Thus, the process $(D_t)$ is a supermartingale under the measure $P_{\mathrm{mix}}$, and its increments are bounded, satisfying $|d_t|\leq L$. 
By the Azuma--Hoeffding inequality we have  

$$
P_{\mathrm{mix}}[D_t\geq b_t]
\leq
\exp\left(-\frac{b_t^2}{2tL^2}\right)
=
\frac{6\delta\lambda_{\min}}{\pi^2 t^2}.
$$
Taking a union bound over all $t\geq 1$ and using the Basel identity $\sum_{u=1}^\infty \frac{1}{u^2}=\frac{\pi^2}{6}$ gives

$$
P_{\mathrm{mix}}[\tau<\infty]
\leq  
\sum_{t=1}^\infty  \frac{6\delta\lambda_{\min}}{\pi^2 t^2} 
=
\delta\lambda_{\min}.
$$
If the true distribution is $q=\theta^k$,  then
$$
P_{\mathrm{mix}}[\tau<\infty]
\geq  
\lambda_0^k P_{\theta^k}[\tau<\infty]
\geq  
\lambda_{\min} P_{\theta^k}[\tau<\infty].
$$
Hence
$$
P_{\theta^k}[\tau < \infty] \leq \delta.
$$
This is exactly the probability of a (false) switch when $q = \theta^k$. With $\delta = 1/T$, this implies that with probability at least $ 1 - 1/T$ the algorithm does not switch.\\

Finally, on the event $\{\tau = \infty\}$ the switching  strategy coincides with the Bayesian strategy, 
so 
$$
R_T(\alpha^{\mathrm{sw}}_{1:T}) = R_T(\alpha^B_{1:T}).
$$
Also since $U_t(\alpha)\in[\log\Delta,0]$, we have $R_T(\cdot)\leq TL$. Therefore,

$$
\E[R_T(\alpha^{\mathrm{sw}}_{1:T})]
\leq   
\E[R_T(\alpha^B_{1:T})]
+  
P_q(\tau<\infty) TL
\leq  
\E[R_T(\alpha^B_{1:T})] + \delta TL. 
$$
With $\delta=1/T$, we have 
$$
\E[R_T(\alpha^{\mathrm{sw}}_{1:T})]
\leq 
\E[R_T(\alpha^B_{1:T})] + L.
$$
The first part of the theorem now follows from Theorem \ref{thm:regret-of-a-Bayesian}.\\ 

\end{proof}
\vspace{5pt}
\section{Arrow Securities}
\label{sec:appendix-Arrow-securities}

Consider an economy with $S$ states and $K$ assets with state-contingent payoffs described by a $K \times S$ matrix of asset payoffs $A$.  A set of assets for this economy is {\bf complete} if for any wealth profile $w=(w_1, \dots, w_S)\in \mathbb{R}^S$, describing wealth in each state, there exists a portfolio of assets, $q=(q^1, \dots, q^K)$, that  achieves $w$. This is clearly equivalent to the requirement that $A$ has rank $S$. So there must be at least $S$ assets, and if there are more than $S$ assets, then $K-S$ of them are redundant. Thus, we restrict attention to the case $K=S$.  

The asset payoff matrix for an economy with a full set of $S$ Arrow securities is the $S \times S$ identity matrix $I$. As the complete market asset payoff matrix $A$ has full rank, its inverse $A^{-1}$ exists and $A^{-1}  A=I$. Note that $A^{-1}$ describes $K$ portfolios of securities; the rows of $A^{-1}$ 
describe
portfolios of the original assets whose payoffs are identical to 
the corresponding Arrow securities. 
Conversely, any asset is payoff-equivalent to the portfolio of Arrow securities described by the payoffs across states of that original asset.  

The prices of Arrow securities are often called ``state prices'' as they are each the price now of one unit of wealth in the state indexing the Arrow security. Given equilibrium prices for Arrow securities the price of any asset is by no-arbitrage the price of the portfolio of Arrow securities that replicate the asset. Similarly, given equilibrium prices of securities, the price of an Arrow security is the price of the portfolio described by $A^{-1}$ that replicates the Arrow security. 

Given this equivalence between security structures we see that for an economy with a complete set of assets, equilibrium can be described equivalently by an equilibrium in the economy with the original assets or equilibrium in an economy with a full set of Arrow securities. A simple example might be useful for seeing what this equivalence requires and what it implies. Suppose, for example, that there are two states and two securities with payoff matrix $A$ given by 
$$
\begin{bmatrix}
1 & 1\\
0 & 1 
\end{bmatrix}
$$

So asset one is a risk free asset and asset two is a risky asset that pays off only in state $2$. In this case $A^{-1}$ is 
$$
\begin{bmatrix}
1 & -1\\
0 & 1 
\end{bmatrix}
$$

The portfolio of original assets that replicates Arrow security one is $(1, -1)$ and the portfolio that replicates Arrow security two is $(0, 1)$. Let $(p_1, p_2)$ be Arrow security prices and $(q_1, q_2)$ be prices of the original securities (with payoff matrix $A$). In an equilibrium these prices must satisfy $p_1 = q_1 - q_2$ and $p_2 = q_2$ as otherwise there is an arbitrage opportunity. Note that this equivalence requires both that the rank of $A$ is $S$ and that agents can hold both long positions and short positions in the original securities; they would never choose to hold short (or even zero) positions in Arrow securities.

\bibliographystyle{abbrvnat} 
\bibliography{learning-in-markets} 

\end{document}